\documentclass[preprint,12pt]{elsarticle}
\biboptions{sort&compress}

\usepackage{amsmath,amssymb,amsthm,mathtools}
\usepackage{graphicx}
\usepackage{booktabs,tabularx,multirow}
\usepackage{subcaption}      % 子图，避免 caption=false 写法
\usepackage{algorithm,algpseudocode}
\usepackage{color}
\usepackage{soul}
\usepackage{cuted}           % 需要跨双栏的长公式/大图可用 strip 环境
\usepackage{hyperref}

\newtheorem{theorem}{Theorem}
\newtheorem{lemma}{Lemma}

\newtheorem{definition}{Definition}
\usepackage{adjustbox} % 放在导言区
\usepackage{float}
\graphicspath{{figure/}}

\journal{Chaos, Solitons \& Fractals}

\begin{document}

\begin{frontmatter}



\title{Structure-Aware Optimal Intervention for Rumor Dynamics on Networks: Node-Level, Time-Varying, and Resource-Constrained}


\author[inst1]{Yan Zhu}
\author[inst1]{Qingyang Liu}
\author[inst1]{Chang Guo}
\author[inst1]{Tianlong Fan\corref{cor1}}
\author[inst1]{Linyuan L\"u\corref{cor1}}

\affiliation[inst1]{organization={School of Cyber Science and Technology, University of Science and Technology of China},
            addressline={}, 
            city={Hefei},
            postcode={230000}, 
            state={Anhui},
            country={China}}

\cortext[cor1]{Corresponding authors. 
Email: tianlong.fan@ustc.edu.cn, linyuan.lv@ustc.edu.cn}

%% Abstract
\begin{abstract}
Rumor propagation in social networks undermines social stability and public trust, calling for interventions that are both effective and resource-efficient. We develop a node-level, time-varying optimal intervention framework that allocates limited resources according to the evolving diffusion state. Unlike static, centrality-based heuristics, our approach derives control weights by solving a resource-constrained optimal control problem tightly coupled to the network structure. Across synthetic and real-world networks, the method consistently lowers both the infection peak and the cumulative infection area relative to uniform and centrality-based static allocations. Moreover, it reveals a stage-aware law: early resources prioritize influential hubs to curb rapid spread, whereas later resources shift to peripheral nodes to eliminate residual transmission. By integrating global efficiency with fine-grained adaptability, the framework offers a scalable and interpretable paradigm for misinformation management and crisis response.
\end{abstract}


% \begin{highlights}
% \item Proposes a node-level adaptive optimal control framework for rumor containment.
% \item Establishes stability condition \( R_e(t)<1 \) ensuring rumor-free equilibrium.
% \item Balances suppression effectiveness and resource efficiency under constraints.
% \item Demonstrates superior performance on both synthetic and real-world networks.
% \end{highlights}

%% Keywords

\end{frontmatter}



\section{\label{sec:Introduction}Introduction}
In today’s highly digitalized and social media-driven era, how to suppress rumor propagation has become a significant challenge affecting social stability, public safety, and individual decision-making~\cite{daley1964epidemics,daley1965stochastic,daniel1973mathematical,zhao2013sir,nekovee2007theory,li2019dynamic,nowzari2016analysis}. 
While social networking platforms enable efficient information dissemination, they simultaneously accelerate the propagation of rumors and misinformation. This enables false information to propagate on a large scale within a short period, causing social panic, economic losses, and even misleading public behavior. Vosoughi et.al~\cite{vosoughi2018spread} indicated that false information spreads significantly faster, deeper, and wider than true information. Consequently, how to effectively suppress rumor propagation has become a focal issue in both academia and the field of social governance. 

In recent years, a wide range of approaches have been explored to suppress rumor propagation. Media- and policy-driven strategies mitigate spreading through external communication, media reporting, silence, or rumor deletion mechanisms~\cite{wang2019stability,cheng2021dynamical,zhu2020stability}. Structural stabilization and probabilistic models emphasize maintaining global stability or threshold conditions under specific assumptions~\cite{zhong2022dynamics,zhu2017dynamical,nekovee2007theory,moreno2004dynamics}. Other studies seek better resource efficiency and realism by considering interest-aware or mobile-network settings~\cite{jeong2018optimal,he2016cost}, saturation incidence and heterogeneous structures~\cite{chen2020dynamical}, and multilingual or structurally diverse environments~\cite{li2020dynamical}. A complementary line disseminates corrective information via influential nodes to counter rumors~\cite{xu2023proactive,yang2019rumor}. Although these works provide valuable insights, most are static or heuristic in nature and struggle to adapt to the temporal evolution of rumor diffusion, which motivates a more principled and adaptive framework based on optimal control.

Building on its demonstrated effectiveness in related domains such as epidemic mitigation~\cite{li2022modeling} and computer-virus intervention~\cite{zhang2016optimal}, optimal control theory has been gradually introduced into rumor suppression. Representative studies incorporate silence mechanisms in delay differential models~\cite{zhu2020delay}, analyze delayed SEIRS-type rumor dynamics with control terms~\cite{xu2010global}, investigate continuous and impulsive vaccination-inspired strategies~\cite{hou2009continuous}, and establish global stability for delayed SEIRS settings~\cite{zhang2008global}. More recent efforts further extend this line by considering nonlinear dynamics with comprehensive interventions~\cite{li2021nonlinear} and stability under heterogeneous networks~\cite{yue2022analysis}. 

% Despite these advances, the vast majority of existing formulations operate at a macroscopic level, adjusting global parameters such as propagation and recovery rates or average intervention intensities. Such aggregate-level interventions overlook the heterogeneity of individual nodes, the temporal evolution of their influence during rumor diffusion, and the critical trade-off between suppression effectiveness and resource efficiency. To address these limitations, we develop a node-level adaptive optimal control framework that dynamically solves per-node intervention weights under explicit resource constraints. By doing so, our approach adapts to temporal shifts in influence, and achieves a principled balance between strong suppression performance and efficient resource usage.

Prior rumor-control studies largely adjust macro-level parameters (propagation/recovery rates or average intensities) or rely on static heuristics, which overlook node heterogeneity, the temporal shift of influence during diffusion, and the effectiveness–cost trade-off. This leaves a gap for node-level, time-varying policies that (i) explicitly encode resource budgets, (ii) adapt to stage transitions during spreading, and (iii) remain interpretable by linking policy to network structure.

\noindent In this paper, we address the above limitations and make the policy design explicit and implementable on real networks by:
\begin{itemize}
  \item formulating a resource-constrained, node-level optimal control of rumor diffusion and deriving explicit optimality conditions for time-varying intervention weights, tightly coupling the policy to network topology;
  \item establishing stability conditions for the controlled dynamics and developing a forward--backward sweep algorithm to compute optimal policies efficiently;
  \item uncovering a robust stage-aware allocation law---early suppression of influential hubs followed by peripheral cleanup---yielding an interpretable, actionable rule for practice;
  \item demonstrating, on diverse synthetic and real networks, that the proposed controller consistently outperforms uniform and centrality-based static baselines in reducing both peak prevalence and cumulative infection burden, achieving a balanced trade-off between effectiveness and resource usage.
\end{itemize}

The remainder of this paper is organized as follows. Section~\ref{sec:Preliminaries} introduces the baseline rumor-propagation model and basic stability analysis. Section~\ref{sec:Methodology} develops the controlled SIR framework and formulates the optimal control problem. Section~\ref{sec:simulation} presents simulation results on synthetic and real networks. Section~\ref{Conclusion} concludes with key findings and practical implications.

\section{\label{sec:Preliminaries}Preliminaries}
The classical susceptible-infected-recovered (SIR) model~\cite{pastor2015epidemic} is employed as the baseline for rumor dynamics. Let $S_i(t)$, $I_i(t)$, $R_i(t)$ denote the probabilities that node $i$ is in the \emph{Susceptible}, \emph{Infected}, and \emph{Recovered} states at time $t$, hereafter abbreviated as \textit{S}, \textit{I}, \textit{R}~\cite{kermack1927contribution}. 
%; after the first mention we write \textit{S}, \textit{I}, \textit{R}. 
The network-based SIR dynamics are governed by%read
\begin{equation}
\label{eq:SIR-uncontrolled}
\begin{aligned}
\dot S_i &= -\,\beta\, S_i \sum_{j=1}^N A_{ij} I_j,\\
\dot I_i &= \beta\, S_i \sum_{j=1}^N A_{ij} I_j - \gamma\, I_i,\\
\dot R_i &= \gamma\, I_i,
\end{aligned}
\end{equation}
where $\beta>0$ is the \emph{propagation rate}, $\gamma>0$ is the \emph{recovery rate}, and $A=(A_{ij})$ is the adjacency matrix of the network.

\subsection{Theoretical Background}
\begin{theorem}[Positivity And Invariance~\cite{hethcote2000mathematics}]
\label{thm:positivity}
For system \eqref{eq:SIR-uncontrolled} with admissible initial data, solutions exist uniquely on $[0,\infty)$ and satisfy
$S_i(t),I_i(t),R_i(t)\in[0,1]$ and $S_i(t)+I_i(t)+R_i(t)=1$ for all $t\ge 0$.
\end{theorem}

\begin{proof}
The right-hand side of~\eqref{eq:SIR-uncontrolled} is locally Lipschitz, ensuring a unique solution. Summing the three equations gives $\frac{d}{dt}(S_i+I_i+R_i)=0$, hence the simplex constraint is invariant. On the boundary $S_i=0$ one has $\dot S_i\ge 0$; similarly $\dot I_i\ge 0$ when $I_i=0$, and $\dot R_i\ge 0$ when $R_i=0$. Thus the positive simplex is forward invariant and each component remains in $[0,1]$~\cite{anderson1991infectious}.
\end{proof}

Linearizing at the rumor-free equilibrium $\mathcal E_0:\, S=\mathbf{1}, I=\mathbf{0}, R=\mathbf{0}$, then obtain the next-generation matrix 
% \[
% \mathcal K_0=\frac{\beta}{\gamma}A,
% \]
\begin{equation}
\mathcal K_0=\frac{\beta}{\gamma}A,
\end{equation}
and the basic reproduction number
\begin{equation}
R_0 = \rho(\mathcal K_0) = \frac{\beta}{\gamma}\rho(A).
\end{equation}

\begin{theorem}[Existence Of Equilibria]
If $R_0<1$, the only equilibrium is the rumor-free equilibrium $\mathcal E_0$; if $R_0>1$, there exists a unique endemic equilibrium $\mathcal E^\ast$ with $I^\ast \gg 0$.
\end{theorem}

\begin{proof}
By Perron-Frobenius~\cite{lajmanovich1976deterministic,diekmann1990definition}, $\rho(A)$ is the simple positive eigenvalue of $A$. The \textit{I}-subsystem dominates removals iff $R_0>1$, implying existence of a positive fixed point. Uniqueness follows from the monotonicity of the vector field in the invariant simplex.
\end{proof}

\begin{lemma}[Spectral Bound of Metzler Jacobian~\cite{horn2012matrix,wang2003epidemic,van2008virus}]
The Jacobian restricted to the infected subspace at $\mathcal E_0$ is a Metzler matrix of the form $\beta A - \gamma I$. 
Its spectral bound is given by $\beta \rho(A) - \gamma$.
\end{lemma}

\begin{theorem}[Local Stability~\cite{van2008virus,wang2003epidemic}]
If $R_0<1$, the rumor-free equilibrium $\mathcal E_0$ is locally asymptotically stable; if $R_0>1$, $\mathcal E_0$ is unstable and $\mathcal E^\ast$ emerges.
\end{theorem}

\begin{proof}
At $\mathcal E_0$, the Jacobian restricted to the \textit{I}-subspace is $\beta A - \gamma I$. Its spectral bound is negative iff $\rho(\beta A / \gamma)<1$, i.e. $R_0<1$.
\end{proof}

These results establish that the reproduction number $R_0$ determines whether rumors vanish or persist. 
In practice, one can alter the effective $\beta,\gamma$ via external intervention. 

The notation employed in this paper is summarized in Table~\ref{tab:notation}, which lists the baseline epidemic variables as well as the additional parameters and variables introduced in the subsequent methodology section.

\begin{table}[H]
\centering
% \captionsetup{font=normalsize} % 让题注用正文字号
\setlength{\tabcolsep}{3pt}
\renewcommand{\arraystretch}{1.05}
\caption{Notation and definitions used in the paper.}
\label{tab:notation}
\renewcommand{\arraystretch}{1}
\normalsize
\begin{tabularx}{\linewidth}{
  >{\raggedright\arraybackslash}p{0.30\textwidth}
  >{\raggedright\arraybackslash}X
}
\toprule
Symbol & Meaning \\
\midrule
$N$ & Number of nodes in the network \\
$A=(A_{ij})$ & Adjacency matrix; $A_{ij}=1$ if nodes $i$ and $j$ are connected \\
$S_i(t)$ & Probability that node $i$ is susceptible at time $t$ \\
$I_i(t)$ & Probability that node $i$ is infected (spreading rumor) at time $t$ \\
$R_i(t)$ & Probability that node $i$ is recovered at time $t$ \\
$\beta$ & Baseline rumor propagation rate (uncontrolled) \\
$\gamma$ & Baseline recovery rate (uncontrolled) \\
$R_0$ & Basic reproduction number without control \\
$\mathcal R_0(w)$ & Controlled basic reproduction number with weights $w_i(t)$ \\
$\rho(A)$ & Spectral radius of the adjacency matrix $A$ \\
$u$ & Global control intensity \\
$w_i(t)$ & Node-level control weight at time $t$ of node $i$ \\
$W_{\text{total}}$ & Total available intervention resources \\
$c$ & Cost coefficient in the objective function \\
$J(w)$ & Objective function (infection prevalence + control cost) \\
$\lambda_{1i}(t), \lambda_{2i}(t), \lambda_{3i}(t)$ & Adjoint variables associated with $S_i, I_i, R_i$ \\
$\lambda_4(t)$ & Lagrange multiplier associated with the resource constraint \\
\bottomrule
\end{tabularx}
\end{table}

\subsection{\label{sec:Problem Definition}Problem Definition}
The above analysis establishes that the basic reproduction number $R_0$ determines whether rumors eventually vanish ($R_0 < 1$) or persist ($R_0 > 1$) in the uncontrolled SIR dynamics. 
This threshold motivates the introduction of \emph{external interventions}, since by modifying the effective propagation rate $\beta$ and recovery rate $\gamma$ one can alter the value of $R_0$ and drive the system toward stability. 

Therefore the definition of the \emph{rumor intervention problem} as follows. 
Consider a network $G=(V,E)$ with adjacency matrix $A$. 
Each node $i$ has state probabilities $S_i(t), I_i(t), R_i(t)$ governed by the SIR dynamics. 
The objective is to design \emph{node-level control variables} $w_i(t)$ that dynamically reallocate limited resources to minimize both rumor prevalence and intervention cost. 

Formally, the rumor intervention problem can be stated as the following optimal control problem~\cite{kirk2004optimal}:

\begin{equation}
\underset{w(t)}{\text{minimize}} \quad 
J(w) = \int_0^T \left( \mathbf{1}^\top I(t) 
+ \tfrac{1}{2c} \, \| w(t) \|_2^2 \right) dt ,
\end{equation}

subject to the controlled SIR dynamics
\begin{equation}
\begin{aligned}
\dot{S}(t) &= -\beta \, S(t) \odot (A I(t)), \\
\dot{I}(t) &= \beta \, S(t) \odot (A I(t)) - \gamma \, I(t), \\
\dot{R}(t) &= \gamma \, I(t),
\end{aligned}
\end{equation}

where $S(t), I(t), R(t) \in \mathbb{R}^N$ denote the vectors of susceptible, infected, and recovered fractions, $A$ is the adjacency matrix, $\odot$ is the Hadamard product.

This formalization provides a precise statement of the control problem and sets the stage for Section~\ref{sec:Methodology}, where the controlled SIR dynamics is derived and solving the corresponding optimal control problem.

\section{\label{sec:Methodology}Methodology}
\subsection{Controlled rumor intervention model}
Based on the threshold analysis in Section~\ref{sec:Preliminaries}, now node-level control is introduced into the SIR model to actively reduce the effective reproduction number and derive optimal strategies for rumor intervention. The node-level intervention is weights $w_i(t)$, scaled by a global control intensity $u$. These weights directly affect both the propagation rate and the recovery rate:
\begin{equation}
    \beta_i(t) = \beta_0 (1 - u w_i(t)), 
    \qquad
    \gamma_i(t) = \gamma_0 (1 + u w_i(t)).
\end{equation}
Here, $\beta_0$ and $\gamma_0$ are propagation rate and recovery rate, while $w_i(t)$ determines the resource allocation to node $i$. 
This design allows for fine-grained interventions that balance suppression effectiveness with overall resource efficiency.

The controlled SIR system on networks is therefore given by:
\begin{equation}
\label{eq:SIR-controlled}
\begin{aligned}
\dfrac{dS_i}{dt} &= -\beta_0 (1 - u w_i(t)) S_i \sum_{j=1}^{N} A_{ij} I_j, \\
\dfrac{dI_i}{dt} &= \beta_0 (1 - u w_i(t)) S_i \sum_{j=1}^{N} A_{ij} I_j - \gamma_0 (1 + u w_i(t)) I_i, \\
\dfrac{dR_i}{dt} &= \gamma_0 (1 + u w_i(t)) I_i, \qquad i=1,\dots,N.
\end{aligned}
\end{equation}

\begin{definition}[Positively Invariant Set]
The positively invariant set of a system refers to the set of initial conditions such that the solution trajectories starting within this set remain within the set for all future times. In the context of the SIR model, the positively invariant set is defined as the set of all states where the susceptible, infected, and recovered fractions satisfy $S_i(t), I_i(t), R_i(t) \in [0, 1],\forall i$, and
\begin{equation}
\sum_{i=1}^N S_i(t) + I_i(t) + R_i(t) = N.
\end{equation}
\end{definition}
This set ensures that the system's state variables remain within the non-negative unit interval, reflecting the realistic constraints of the problem.

\begin{lemma}[Positivity~\cite{kermack1927contribution,hethcote2000mathematics}]
\label{lem:positivity}
In the node-level adaptive optimal control SIR model, the solutions \( S_i(t) \), \( I_i(t) \), and \( R_i(t) \) remain positive for all \( t > 0 \) and \( i = 1, 2, \dots, N \), and satisfy \( S_i(t) + I_i(t) + R_i(t) = 1 \) for all \( t \geq 0 \). The solutions are unique for all \( t \geq 0 \).
\end{lemma}

\begin{proof}
We assume that the initial state of the system is positive, \( S_i(0) > 0 \), \( I_i(0) > 0 \), \( R_i(0) > 0 \), for all \( i \). These initial conditions guarantee that each node's state is positive at \( t = 0 \). The node-level intervention is introduced via the control variables \( w_i(t) \) to adjust the propagation rate \( \beta_i(t) \) and recovery rate \( \gamma_i(t) \). The state dynamics of each node are governed by the following system of equations~\ref{eq:SIR-controlled}. Here, \( u \) is the global control intensity, and \( w_i(t) \) is the intervention weight for node \( i \). Based on the existence and uniqueness theorem for ordinary differential equations~\cite{coddington1955,perko2001}, the solutions to the system are continuous. Therefore, there exists some time \( t_1 > 0 \) such that for all nodes \( i \),
\begin{equation}
S_i(t) > 0, \quad I_i(t) > 0, \quad R_i(t) > 0, \quad \forall t \in [0, t_1).
\end{equation}
In other words, all state variables remain positive up to time \( t_1 \). Suppose there exists a node \( k_1 \) such that at time \( t_1 \), \( S_{k_1}(t_1) = 0 \), meaning the node is no longer susceptible. According to the dynamics, we have \( \dot{S_{k_1}} = 0 \), and hence \( S_{k_1}(t) \) cannot become negative~\cite{hethcote2000mathematics}, as \( S_{k_1}(t) \geq 0 \). Similarly, if \( I_{k_1}(t_1) = 0 \), then \( \dot{I_{k_1}} \geq 0 \), and \( I_{k_1}(t) \) will not become negative. Thus, by contradiction, we conclude that for all \( t > 0 \), the state variables \( S_i(t) \), \( I_i(t) \), and \( R_i(t) \) will remain positive and satisfy the total sum constraint \( S_i(t) + I_i(t) + R_i(t) = 1 \) for all \( t \geq 0 \).
\end{proof}

For constant \( w_i(t) \), the time-frozen next-generation matrix, which encapsulates the transmission dynamics of the network, is given by:
\begin{equation}
\mathcal K(w) = \frac{\beta_0}{\gamma_0} \, \mathrm{diag}\!\left(\frac{1 - u w_i}{1 + u w_i}\right) A,
\end{equation}
where \( \mathcal K(w) \) captures how the interventions modify the network’s contact structure by adjusting the transmission dynamics for each node \( i \). The controlled reproduction number, which reflects the overall effect of these interventions on the network’s ability to sustain the epidemic, is defined as the spectral radius of \( \mathcal K(w) \):
\begin{equation}
\mathcal R_0(w) = \rho(\mathcal K(w)),
\end{equation}
where \( \rho(\mathcal K(w)) \) represents the largest eigenvalue of the matrix \( \mathcal K(w) \), which serves as a threshold for epidemic persistence. Specifically, if \( \mathcal R_0(w) > 1 \), the rumor will persist in the network; whereas if \( \mathcal R_0(w) < 1 \), the rumor will die out.

\begin{theorem}[Existence Under Control]
If \( \mathcal R_0(w) < 1 \), the only equilibrium is the rumor-free equilibrium. If \( \mathcal R_0(w) > 1 \), there exists a unique endemic equilibrium with \( I^\ast \gg 0 \).
\end{theorem}

\begin{proof}
We begin by linearizing the system at the Disease-Free Equilibrium (DFE). This linearization leads to the Metzler matrix:
\begin{equation}
M(w) = B(w) A - \Gamma(w),
\end{equation}
where \( B(w) \) represents the effect of the intervention on the propagation rate, \( A \) is the adjacency matrix of the network, and \( \Gamma(w) \) captures the effect on the recovery rate. By the Perron-Frobenius theorem~\cite{van2002reproduction}, we know that \( M(w) \) is Hurwitz (all eigenvalues have negative real parts) if and only if:
\begin{equation}
\rho(\Gamma(w)^{-1} B(w) A) < 1.
\end{equation}
The argument from monotone systems then implies the dichotomy: if \( \mathcal R_0(w) < 1 \), the rumor-free equilibrium is the only equilibrium, and if \( \mathcal R_0(w) > 1 \), a unique endemic equilibrium exists with a substantial infection \( I^\ast \gg 0 \)~\cite{van2002reproduction, shuai2013global}.
\end{proof}

\subsection{Stability Analysis Under Control}

\begin{theorem}[Local Stability With Control]
If \( \mathcal R_0(w) < 1 \), the rumor-free equilibrium is locally asymptotically stable; if \( \mathcal R_0(w) > 1 \), it is unstable~\cite{wang2003epidemic, van2008virus}.
\end{theorem}

\begin{theorem}
Let \( w_i(t) \) be piecewise-continuous. Define the instantaneous effective reproduction number as:
\begin{equation}
R_e(t) = \frac{\beta_0}{\gamma_0}\, \rho\!\left(\mathrm{diag}\!\left(\frac{1 - u w_i(t)}{1 + u w_i(t)}\right) A\right).
\end{equation}
If \( \sup_{t \in [0, T]} R_e(t) < 1 \), then the rumor-free equilibrium is uniformly exponentially stable on \( [0, T] \).
\end{theorem}

\begin{proof}
Freeze \( w(t) \) to obtain \( \dot I = M(t) I \), where \( M(t) \) is the Metzler matrix that evolves with time. The condition \( \sup_t R_e(t) < 1 \) implies that each \( M(t) \) is Hurwitz with a uniform bound, meaning all the eigenvalues of \( M(t) \) have negative real parts. Applying Grönwall’s inequality~\cite{gronwall1919note} yields exponential decay~\cite{khanafer2016stability, nowzari2016analysis, hethcote2000mathematics}, ensuring that the rumor-free equilibrium remains stable over time.
\end{proof}

\subsection{Research on the optimal control problem}
While stability conditions guarantee whether rumors vanish or persist under given controls, in practice one seeks to determine time-varying control strategies that optimally balance suppression and cost. We therefore formulate the following optimal control problem.

The optimal control problem is formulated as minimizing the total infection prevalence and intervention cost:
\begin{equation}
J(w) = \int_0^T \left( \sum_{i=1}^N I_i(t) + \tfrac{1}{2} c \sum_{i=1}^N w_i(t)^2 \right) dt,
\end{equation}
subject to resource constraints
\begin{equation}
\sum_{i=1}^N w_i(t) \leq W_{\text{total}}, 
\qquad 
w_i(t) \geq 0, \quad \forall t \in [0,T].
\end{equation}

It is worth noting that, unlike existing methods which regulate global-level control parameters, our formulation explicitly introduces node-level control variables $w_i(t)$ into the SIR dynamics. This design enables fine-grained interventions that adapt to the evolving influence of individual nodes. At the same time, the quadratic control cost in the objective function ensures that resource usage is penalized, thereby embedding a principled trade-off between suppression effectiveness and intervention efficiency directly into the optimization process.

\begin{lemma}[Existence of Optimal Control~\cite{zhang2016optimal,lenhart2007optimal}]
\label{lem:existence-optimal}
For the controlled system~\eqref{eq:SIR-controlled} under admissible initial conditions, there exists an optimal control 
% \[
% w^{\ast}(t) = \big(w^{\ast}_1(t), w^{\ast}_2(t), \dots, w^{\ast}_n(t)\big),
% \]
% such that 
% \[
% J(w^{\ast}(t)) = \min_{w(t)\in U} J(w(t)).
% \]
\begin{equation}
w^{\ast}(t) = \big(w^{\ast}_1(t), w^{\ast}_2(t), \dots, w^{\ast}_N(t)\big),
\end{equation}
such that 
\begin{equation}
J(w^{\ast}(t)) = \min_{w(t)\in U} J(w(t)),
\end{equation}
where the admissible control set $U$ is defined as
\begin{equation}
U = \Bigg\{\, w(t) \in L^{\infty}([0,T]; \mathbb{R}^N) \;\Bigg|\;
\begin{aligned}
& \sum_{i=1}^N w_i(t) \leq W_{\text{total}}, \\ 
& w_i(t) \geq 0, \;\; \forall t \in [0,T]
\end{aligned}
\Bigg\}.
\end{equation}

\end{lemma}

Once the existence of optimal controls has been established, the necessary conditions for optimality are obtained via Pontryagin’s Maximum Principle~\cite{lenhart2007optimal}.

\begin{theorem}\label{thm:SIR-optimal}
Let $(S^*(t), I^*(t), R^*(t))$ be the optimal state trajectories corresponding to the optimal control $w_i^*(t)$ in the controlled SIR network system~\eqref{eq:SIR-controlled}. Then there exist adjoint variables $\lambda_{1i}(t), \lambda_{2i}(t), \lambda_{3i}(t)$ and a Lagrange multiplier $\lambda_4(t)$~\cite{lagrange1813theorie} such that
\begin{equation}
\left\{
\begin{aligned}
\dfrac{d\lambda_{1i}(t)}{dt} &= (\lambda_{1i}(t)-\lambda_{2i}(t))\, \beta_0 (1 - u w_i^*(t)) \sum_{j=1}^{N} A_{ij} I_j^*(t), \\[6pt]
\dfrac{d\lambda_{2i}(t)}{dt} &= -1 + \lambda_{2i}(t)\gamma_0(1+u w_i^*(t)) \\
&\quad - (\lambda_{1i}(t)-\lambda_{2i}(t))\beta_0 (1 - u w_i^*(t)) S_i^*(t) \sum_{j=1}^{N} A_{ij} \\
&\quad - \lambda_{3i}(t)\gamma_0(1+u w_i^*(t)), \\[6pt]
\dfrac{d\lambda_{3i}(t)}{dt} &= 0, \\[6pt]
\dfrac{d\lambda_4(t)}{dt} &= -\sum_{i=1}^{N} w_i^*(t) + W_{\text{total}},
\end{aligned}
\right.
\end{equation}
with transversality conditions

\begin{equation}
\lambda_{1i}(T) = \lambda_{2i}(T) = \lambda_{3i}(T) = 0, \quad i=1,2,\dots,N.
\end{equation}

Moreover, the optimal control $w_i^*(t)$ has the explicit form
\begin{equation}
\begin{aligned}
w_i^*(t) &= \max\Biggl\{0,\; \min\Biggl\{ 
   -\frac{1}{c}\Bigl[ \beta_0 u S_i^*(t)\!\!\sum_{j=1}^{N} A_{ij} I_j^*(t)(\lambda_{1i}(t) 
   \\
   &\qquad - \lambda_{2i}(t)) 
   - \gamma_0 u I_i^*(t)(\lambda_{2i}(t) - \lambda_{3i}(t)) 
   + \lambda_4(t) \Bigr],\,1 
   \Biggr\}\Biggr\}.
\end{aligned}
\end{equation}

\end{theorem}

\begin{proof}
The proof follows from Pontryagin’s Maximum Principle~\cite{lenhart2007optimal}. The Hamiltonian is constructed as follows:
\begin{equation}
\begin{aligned}
H ={}& \sum_{i=1}^N I_i(t) + \frac{1}{2} c \sum_{i=1}^N (u w_i(t))^2 + \sum_{i=1}^N \Biggl[ \lambda_{1i}\biggl(-\beta_0 (1 - u w_i) S_i \sum_{j=1}^N A_{ij} I_j \biggr) \\
&+\lambda_{2i}\biggl(\beta_0 (1 - u w_i) S_i\sum_{j=1}^N A_{ij} I_j - \gamma_0 (1 + u w_i) I_i\biggr)\\ 
&+ \lambda_{3i}\bigl(\gamma_0 (1 + u w_i) I_i\bigr) \Biggr] 
+ \lambda_4 \biggl(\sum_{i=1}^N w_i - W_{\text{total}}\biggr).
\end{aligned}
\end{equation}

To derive the necessary conditions for optimality, we differentiate the Hamiltonian \( H \) with respect to the state variables \( S_i \), \( I_i \), and \( R_i \) to obtain the adjoint system. The first-order conditions for optimality lead to the following system of differential equations for the adjoint variables.
Adjoint equation for \( \lambda_{1i}(t) \):Differentiating the Hamiltonian with respect to \( S_i \), we get
\begin{equation}
\begin{aligned}
\frac{\partial H}{\partial S_i} &= -\lambda_{1i}(t) \beta_0 (1 - u w_i(t)) \sum_{j=1}^N A_{ij} I_j(t) + \lambda_{2i}(t) \beta_0 (1 - u w_i(t)) \sum_{j=1}^N A_{ij} I_j(t).
\end{aligned}
\end{equation}
which leads to the adjoint equation:
\begin{equation}
\frac{d\lambda_{1i}(t)}{dt} = (\lambda_{1i}(t) - \lambda_{2i}(t)) \beta_0 (1 - u w_i^*(t)) \sum_{j=1}^N A_{ij} I_j^*(t).
\end{equation}

This equation describes how the adjoint variable \( \lambda_{1i}(t) \), associated with the infected state \( I_i(t) \), evolves over time.

Adjoint equation for \( \lambda_{2i}(t) \):Similarly, differentiating the Hamiltonian with respect to \( I_i \), we get
\begin{equation}
\begin{aligned}
\frac{\partial H}{\partial I_i} &= \lambda_{1i}(t) \beta_0 (1 - u w_i(t)) S_i(t) \sum_{j=1}^N A_{ij} I_j(t) \\
& + \lambda_{2i}(t) \beta_0 (1 - u w_i(t)) S_i(t) \sum_{j=1}^N A_{ij} I_j(t) \\
& - \lambda_{2i}(t) \gamma_0 (1 + u w_i(t)) I_i(t).
\end{aligned}
\end{equation}
which leads to the adjoint equation:
\begin{equation}
\begin{aligned}
\frac{d\lambda_{2i}(t)}{dt} &= -1 + \lambda_{2i}(t) \gamma_0 (1 + u w_i^*(t)) - (\lambda_{1i}(t) - \lambda_{2i}(t)) \\
&\beta_0 (1 - u w_i^*(t)) S_i^*(t) \sum_{j=1}^N A_{ij} - \lambda_{3i}(t) \gamma_0 (1 + u w_i^*(t)).
\end{aligned}
\end{equation}

Adjoint equation for \( \lambda_{3i}(t) \):Differentiating the Hamiltonian with respect to \( R_i \), we get
\begin{equation}
\frac{\partial H}{\partial R_i} = \lambda_{3i}(t) \gamma_0 (1 + u w_i(t)),
\end{equation}
which gives the adjoint equation:
\begin{equation}
\frac{d\lambda_{3i}(t)}{dt} = 0.
\end{equation}

Finally, the first-order derivative of \( H \) with respect to \( w_i \) is:
\begin{equation}
\begin{aligned}
\frac{\partial H}{\partial w_i} &= \sum_{i=1}^N \left[ \lambda_{1i}(t) \left( \beta_0 (1 - u w_i(t)) S_i(t) \sum_{j=1}^N A_{ij} I_j(t) \right) \right. \\
& \quad  - \lambda_{2i}(t) \gamma_0 (1 + u w_i(t)) I_i(t) \Bigr],
\end{aligned}
\end{equation}

which leads to the optimal control equation: \( \frac{\partial H}{\partial w_i} = 0 \), so
\begin{equation}
\begin{aligned}
w_i^*(t) &= \max \left\{ 0, \min \left\{ - \frac{1}{c} \left[ \beta_0 u S_i^*(t) \sum_{j=1}^N A_{ij} I_j^*(t) (\lambda_{1i}(t) \right. \right. \right. \\
& \quad - \lambda_{2i}(t)) - \gamma_0 u I_i^*(t) (\lambda_{2i}(t) - \lambda_{3i}(t)) + \lambda_4(t) \Bigr], 1 \Biggr\}\Biggr\}.
\end{aligned}
\end{equation}

The transversality conditions follow from the fact that at terminal time \( T \), the adjoint variables \( \lambda_{1i}(T) \), \( \lambda_{2i}(T) \), and \( \lambda_{3i}(T) \) must all be zero, as there is no terminal cost associated with the state variables:
\begin{equation}
\lambda_{1i}(T) = \lambda_{2i}(T) = \lambda_{3i}(T) = 0, \quad i = 1, 2, \dots, N.
\end{equation}

Thus, we have derived the necessary conditions for the optimal control, and the explicit form of the optimal control \( w_i^*(t) \) is given by the equation above.
\end{proof}

\begin{algorithm}[H]
\caption{Forward--Backward Sweep Algorithm.}
\label{alg:generic-fbs}
\begin{algorithmic}[1]
\Require Initial state $y_0$, initial control $w^{(0)}(t)$, time grid $\{t_k\}_{k=0}^{K}$ on $[0,T]$, tolerance $\varepsilon$, smoothing $\tau\in(0,1]$, max iterations $M$
\Ensure Control $w^\star(t)$
\State \textbf{Initialization:} set $m\gets 0$; compute initial state trajectory under $w^{(0)}$; set terminal adjoints (e.g., $\lambda(T)=0$)
\Repeat
  \State \textbf{Forward sweep:} with $w^{(m)}$, integrate the state ODEs on $[t_k,t_{k+1}]$, $k=0,\dots,K-1$, to obtain $y^{(m)}(t)$
  \State \textbf{Backward sweep:} set $\lambda_1(T)=\lambda_2(T)=\lambda_3(T)=0$; integrate the adjoint ODEs backward for $k=K-1,\dots,0$ to obtain $\lambda^{(m)}(t)$
  \State \textbf{Optimality step:} for each $t_k$, solve the first–order condition
        \[
          \frac{\partial H}{\partial w}\big(y^{(m)}(t_k),\lambda^{(m)}(t_k),w\big)=0
        \]
        to get a candidate $\tilde w(t_k)$ (or take a gradient step toward minimizing the Hamiltonian)
  \State \textbf{Update:} optionally relax $w^{(m+1)}(t_k)\leftarrow \tau\,w^{(m)}(t_k)+(1-\tau)\,\tilde w(t_k)$ with $\tau\in(0,1)$
  \State \textbf{Stopping test:} if
        \[
           \frac{\|w^{(m+1)}-w^{(m)}\|}{\|w^{(m)}\|} < \varepsilon
        \]
        then \textbf{break}; else set $m\gets m+1$
\Until{$m=M$}
\State \Return $w^\star(t)\gets w^{(m)}(t)$
\end{algorithmic}
\end{algorithm}

Based on the above theory, the forward-backward scanning Algorithm~\ref{alg:generic-fbs} is used to optimize the control strategy and reduce the infection density. 

By solving the system of differential equations governing the SIR model and applying the control strategies iteratively, the algorithm determines the optimal control input for each time step. This approach is designed to minimize the infection spread while adhering to resource constraints. The algorithm can be applied with different control strategies to compare their performance in controlling the epidemic.

\begin{table}[H]
    \centering
    \captionsetup{font=small} 
    \caption{Basic statistics of the real datasets. Here, \(N\) and \(M\) denote the total number of nodes and edges, respectively, \(\langle k \rangle\) and \(\langle L \rangle\) represent the mean degree and mean shortest path length, \(C\) is the mean clustering coefficient~\cite{watts1998collective}, \(r\) is the degree assortativity coefficient, and \(H\) is the heterogeneity index (coefficient of variation of degree distribution).}
    \setlength{\tabcolsep}{1pt}
    \renewcommand{\arraystretch}{1.05}
    \normalsize
    \begin{tabularx}{\linewidth}{ccccccccc}
    \toprule
    Dataset &  Abbreviation& \(N\) & \(M\) & \(\langle k \rangle\) & \(C\) & \(\langle L \rangle\) & \(r\) & \(H\) \\
    \midrule
    BA\_network & BA&500 & 1491 & 5.9640 & 0.0466 & 3.261 & -0.0624 & 1.0416\\
    WS\_network & WS&500 & 2500 & 10.0 & 0.4997 & 4.0035 & -0.0151 & 0.0990\\
    ER\_network &ER & 500 & 6125 & 24.5 & 0.0493 & 2.2381 &-0.0043&0.2001\\
    Dolphins & DP&62 & 159 & 5.129 & 0.259 & 3.357 &-0.0436&0.5716\\
    Karate & KT&34 & 78 & 4.5882 & 0.5706 & 2.4082&-0.0456&0.8326 \\
    % Windsurfers&WF & 43 & 336 & 15.6279 & 0.6534 & 1.6711 \\
    
    Sociopatterns &SC& 410 & 2765 & 13.4878 & 0.4558 & 3.6309 & 0.2258 & 0.6226\\
    Facebook &FB& 1266 & 6451 & 10.1912 & 0.0684 & 3.3103 &-0.0844&1.2989\\
    Email & EM& 1133 & 5451 & 9.6222 & 0.2202 & 3.606 & 0.0782 & 0.9706\\
    Digg & DG& 30398 & 86312 & 5.6788 & 0.0053 & 4.6731 & 0.0083 & 1.9883\\
    Enron & EN& 36692 & 183831 & 10.0202 & 0.4970 & 4.0252 & -0.1108 & 3.6027\\
    BlogCatalog3 & BL& 10312 & 333983 & 64.7756 & 0.4632 & 2.3824 & -0.2541 & 2.7433\\

    \bottomrule
\end{tabularx}
\label{tab:network_topology}
\end{table}

\section{\label{sec:simulation}Simulation and Analysis}
To comprehensively evaluate the proposed adaptive optimal control strategy, three complementary experiments are desined from the perspectives of (i) micro-level correlations, (ii) weight allocation patterns, and (iii) macro-level spreading dynamics. These experiments together provide a systematic understanding of how the strategy allocates resources and suppresses rumor propagation.

The datasets used in the experiment are as follows. Two canonical synthetic topologies with $N=500$ nodes each are generated. The Barabási–Albert (BA) scale-free model~\cite{barabasi1999emergence} captures hub-dominated, power-law degree distributions typical of online platforms, providing a heterogeneous baseline where a few nodes possess disproportionate connectivity. The Watts–Strogatz (WS) small-world model~\cite{watts1998collective} exhibits high clustering and short average path lengths, reflecting the local cohesion and efficient navigation observed in many social systems. The Erdős–Rényi (ER) random-graph model~\cite{erdds1959random,erdHos1960evolution} providing a homogeneously mixed baseline in which edges are placed independently with $p$ probability.

Seven real networks spanning social, communication, and face-to-face contact datasets are considered. Zachary's Karate Club network~\cite{zachary1977information}, a classic benchmark for community detection; the Lusseau's Dolphins social network~\cite{lusseau2003bottlenose}, illustrating small-world properties in animal societies; the Génois et al. SocioPatterns dataset~\cite{genois2015data}, featuring high-resolution temporal contact data for epidemic modeling. The Ia-fb-messages network~\cite{rossi2015network}, a Facebook-like social network derived from an online community for students at the University of California, Irvine, comprising users who have sent or received at least one message. The organizational Email network~\cite{guimera2003self} reflects internal email communications with meso-scale and hierarchical features; an undirected projection is used to capture communication ties. The Digg interaction network~\cite{de2009social} encodes user interactions on a social news platform; a static, undirected projection of user–user interactions is constructed. The BlogCatalog3 friendship network~\cite{leskovec2009community}, a blogging community graph where nodes represent blogger accounts and edges denote undirected friendship ties, with the edges file listing unweighted pairs and auxiliary files providing group-membership information. Finally, the Enron email communication network~\cite{tang2009relational}, depicting organizational email exchanges in which nodes correspond to email addresses and edges capture directed communications, with the edges file recording and, in some releases, timestamps or aggregated message counts.

These networks differ in size, clustering, and degree heterogeneity, which allows us to test the robustness and adaptability of our control strategy across diverse topologies. The specific information of each network dataset is shown in the following table~\ref{tab:network_topology}. These statistics include the number of nodes \(N\), the number of links \(M\), the average degree \(\langle k \rangle\), average clustering coefficient \(C\), mean shortest path length \(\langle L \rangle\), and we also report two important network characteristics: degree assortativity (\(r\)) and heterogeneity index (\(H\)). 

The degree assortativity coefficient \(r\) quantifies the tendency of nodes to connect to others of similar degree. It is defined as the Pearson correlation between the degrees at the two ends of a uniformly random edge~\cite{newman2002assortative}:
\begin{equation}
r \;=\;
\frac{\operatorname{Cov}\!\bigl(K^{(1)},K^{(2)}\bigr)}
     {\sqrt{\operatorname{Var}\!\bigl(K^{(1)}\bigr)\,
            \operatorname{Var}\!\bigl(K^{(2)}\bigr)}},
\end{equation}
where \(K^{(1)}\) and \(K^{(2)}\) denote the degrees of the two endpoints of a randomly chosen edge. Positive \(r\) indicates assortative mixing (high-degree nodes tend to connect to other high-degree nodes), whereas negative \(r\) indicates disassortative mixing (high-degree nodes tend to connect to low-degree nodes).

The heterogeneity index \(H\) quantifies the unevenness of degree distribution and is calculated as the coefficient of variation of the degree distribution:
\begin{equation}
H = \frac{\sigma(k)}{\mu(k)} = \frac{\sqrt{\frac{1}{N}\sum_{i=1}^{N}(k_i - \mu(k))^2}}{\frac{1}{N}\sum_{i=1}^{N}k_i}
\label{eq:heterogeneity}
\end{equation}
where \(\sigma(k)\) is the standard deviation of the degree distribution, \(\mu(k)\) is the mean degree, and \(N\) is the total number of nodes. Higher values of \(H\) indicate greater heterogeneity in node connectivity patterns, with scale-free networks typically exhibiting high heterogeneity due to the presence of hub nodes~\cite{albert2002statistical}.

\subsection{\label{sec:correlation}Correlation Between Node Weights and Centrality}
The first experiment aims to investigate the relationship between the control weights $w_i$ assigned to each node and their structural importance in the network. Specifically, we select three widely used centrality measures: degree centrality (DC)~\cite{newman2018networks}, betweenness centrality (BC)~\cite{freeman1977set}, closeness centrality (CC)~\cite{freeman1978centrality}.

To facilitate the interpretation of the weight-center relationship, we select several smaller networks in the dataset (karate, Dolphin, the small Email network, Sociopatterns, and two synthetic networks) in section~\ref{sec:correlation} and section~\ref{sec:visualization}. Among them, the node roles are easier to check visually and statistically.

To quantitatively measure the correlations between $w_i$ and these metrics, we employ the Pearson correlation coefficient~\cite{benesty2009pearson}, a widely used measure of linear dependence between two variables. The coefficient is defined as:
\begin{equation}
r = \frac{\sum_{i=1}^N (x_i - \bar{x})(y_i - \bar{y})}
         {\sqrt{\sum_{i=1}^N (x_i - \bar{x})^2}\sqrt{\sum_{i=1}^N (y_i - \bar{y})^2}},
\label{pearson}
\end{equation}
where $x_i$ and $y_i$ represent the paired observations (here, centrality values and control weights), $\bar{x}$ and $\bar{y}$ denote their sample means, and $N$ is the number of nodes. 

By computing the Pearson correlation coefficient for each network and each centrality metric, we can track how the alignment between structural importance (as measured by centrality) and resource allocation (as reflected in $w_i(t)$) evolves over time. This allows us to reveal network-specific patterns---for instance, whether high-degree hubs~\cite{shi2025cycrank,fan2017generalization} in scale-free networks or bridging nodes in social networks dominate the allocation in different propagation stages.

\begin{figure*}[htbp]
\centering

% Row 1
\subfloat[DP\label{fig:dolphins_corr}]{
  \includegraphics[width=0.31\textwidth]{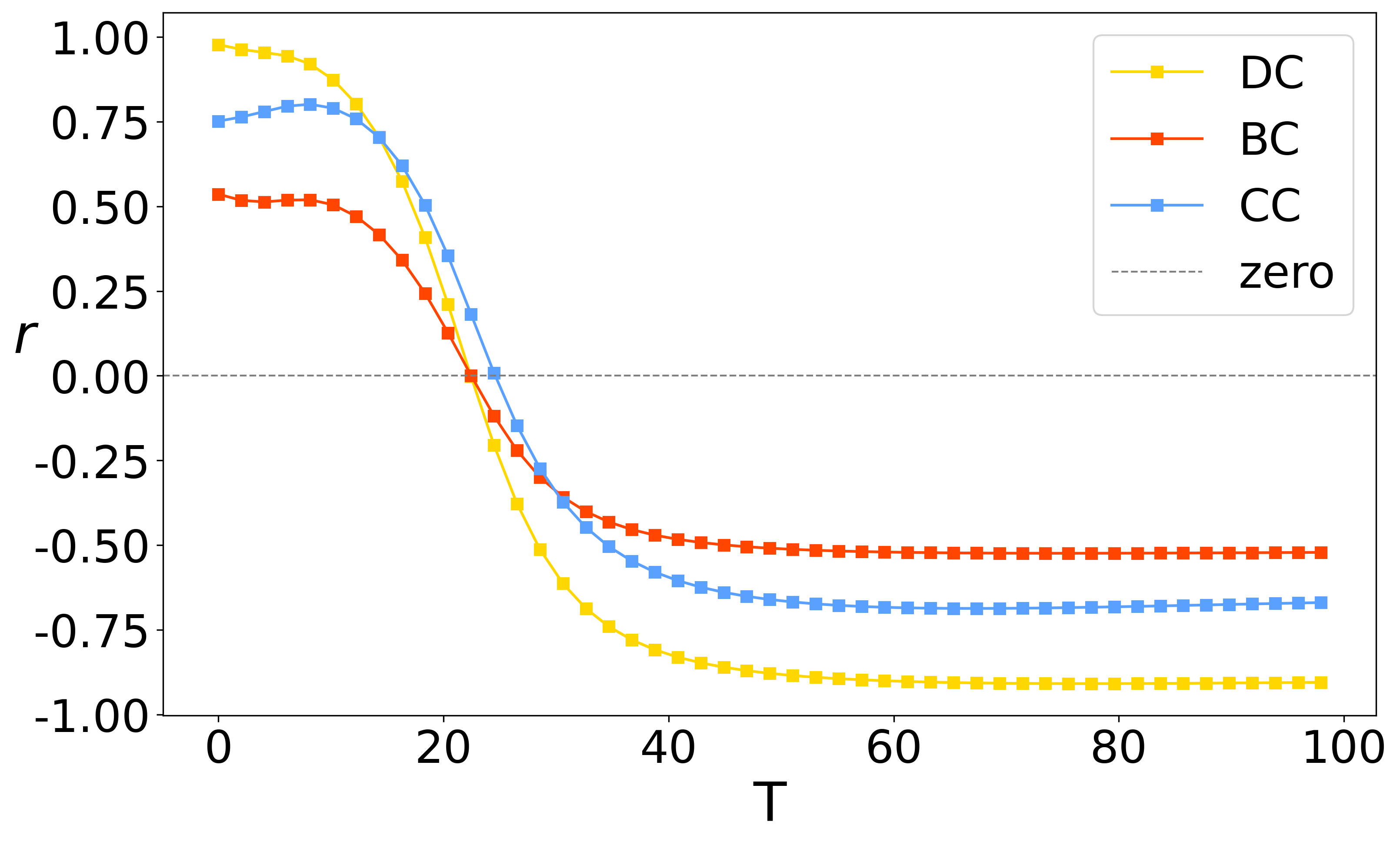}
}\hfill
\subfloat[KT\label{fig:karate_corr}]{
  \includegraphics[width=0.31\textwidth]{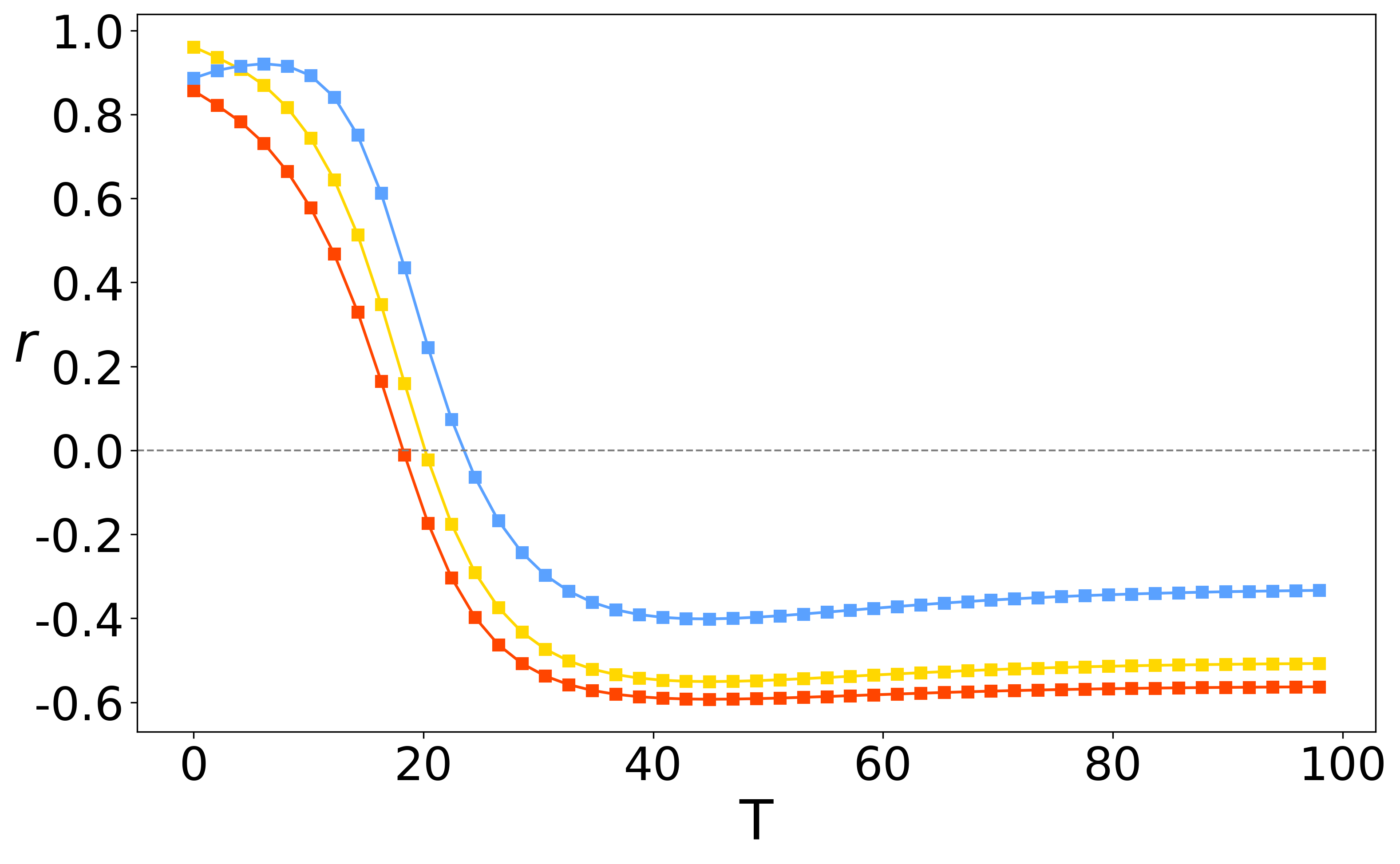}
}\hfill
\subfloat[EM\label{fig:Email_corr}]{
  \includegraphics[width=0.31\textwidth]{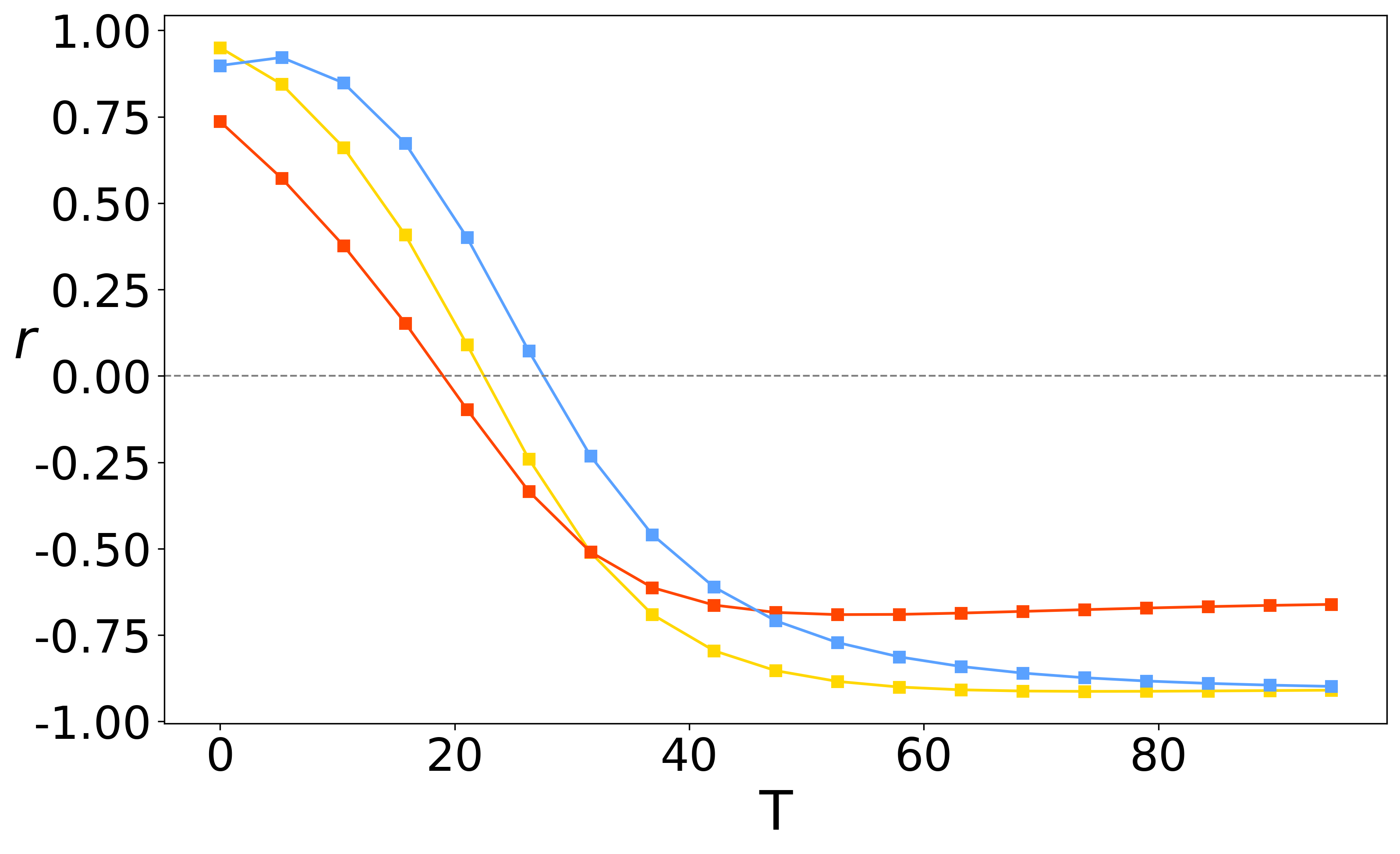}
}

\par\medskip % ← 换行，开始第二行

% Row 2
\subfloat[SC\label{fig:sociopatterns_corr}]{
  \includegraphics[width=0.31\textwidth]{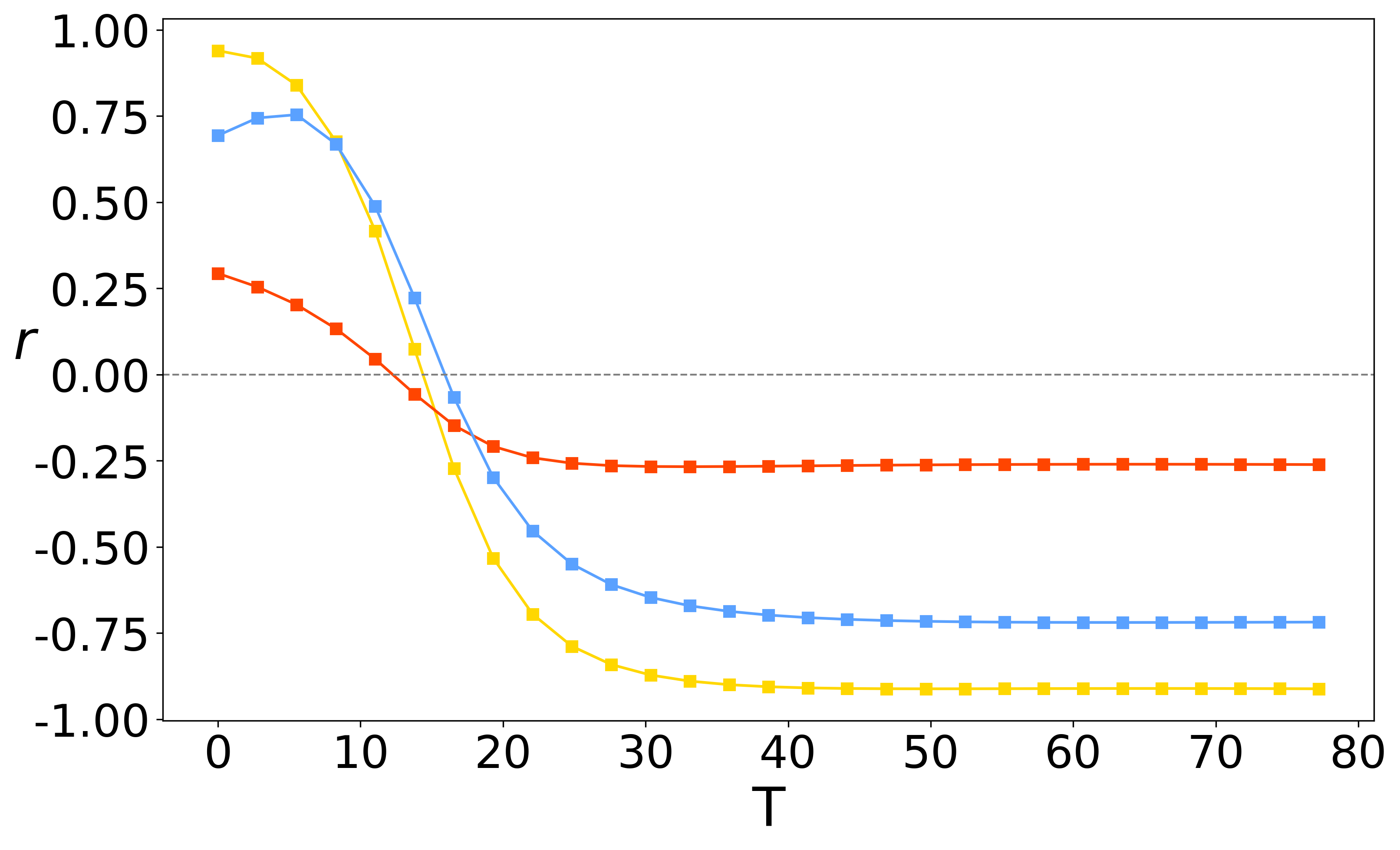}
}\hfill
\subfloat[BA\label{fig:ba_corr}]{
  \includegraphics[width=0.31\textwidth]{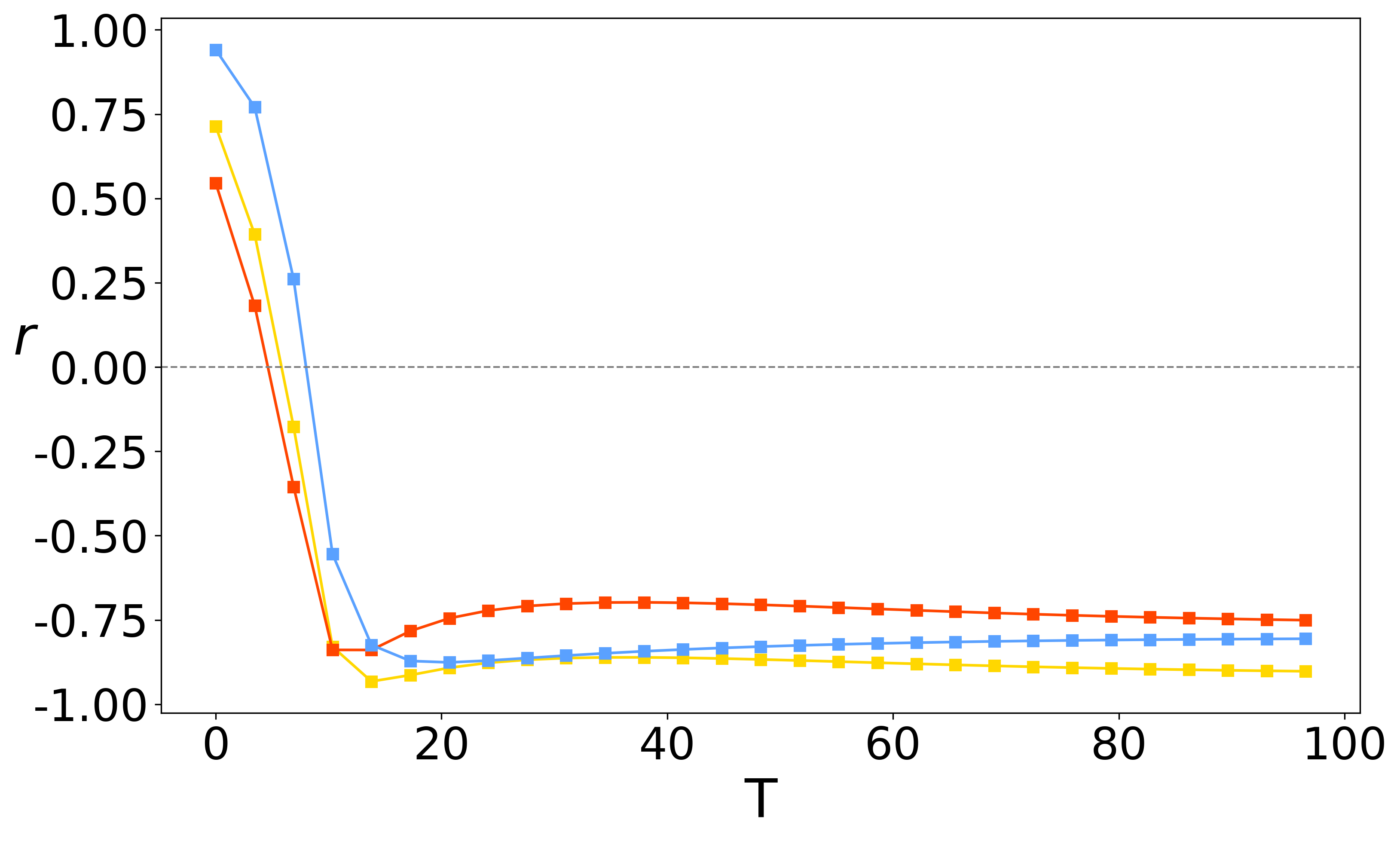}
}\hfill
\subfloat[WS\label{fig:ws_corr}]{
  \includegraphics[width=0.31\textwidth]{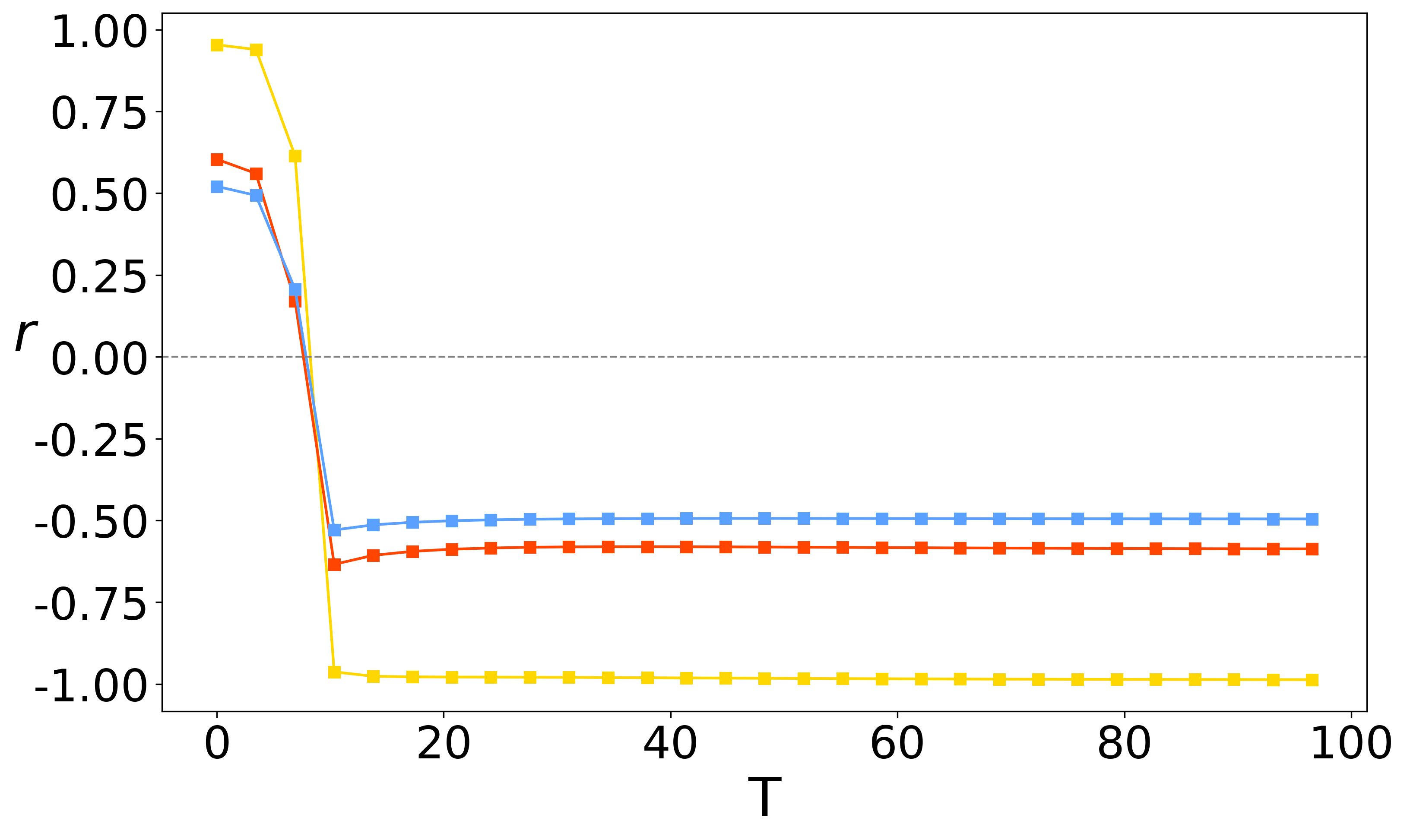}
}

\caption{Temporal correlation between control Weights and network static centralities.}
\label{fig:corr}
\end{figure*}

Fig.~\ref{fig:corr} shows the temporal evolution of the Pearson correlation coefficient $r$ between the control weight $w_i(t)$ and three centrality measures (DC, BC, and CC) across different networks. The results highlight that during the process of rumor dissemination, the correlation between the weights obtained by nodes and the centrality of nodes has undergone a significant transformation.  

In the early stage of propagation, all networks exhibit strong positive correlations, indicating that control resources are mainly allocated to high-centrality nodes~\cite{qiu2021identifying,shi2023cost}. This strategy rapidly suppresses the influence of hubs, bridges, and highly accessible nodes, preventing them from dominating the initial spread. During the middle stage, correlations drop sharply. The correlation with DC and CC becomes negative, while BC decreases more slowly and stabilizes at a moderate negative value. This suggests a shift toward protecting low-centrality nodes, which otherwise could serve as secondary sources of rumor transmission. In the late stage of propagation, the correlation patterns stabilize. DC and CC remain strongly negative, while BC stays weakly to moderately negative. This indicates that resources are persistently allocated to low-centrality nodes, while bridge nodes still retain partial weight due to their residual role in connecting communities. Meanwhile, across different network topologies, the same trend is observed.

\subsection{\label{sec:visualization}Visualization of Weight Allocation by Centrality}
To examine the weighting scheme during the propagation process, we visualized the temporal evolution of node weights based on normalized degree centrality.
%To explore the weights assigned to each node during the propagation process, we visualized the magnitudes of the weights assigned to each node in the network over time based on the normalized degree centrality.

Therefore, we conducted a second experiment, which focused on visualizing how the intervention weights are distributed among nodes of different structural roles during rumor propagation. For clarity, we use degree centrality as a representative metric and plot the temporal evolution of weights for nodes with different degree classes. The visualization reveals a stage-specific pattern.

\begin{figure*}[htbp]
\centering

% Row 1
\subfloat[DP\label{fig:dolphins}]{
  \includegraphics[width=0.31\textwidth]{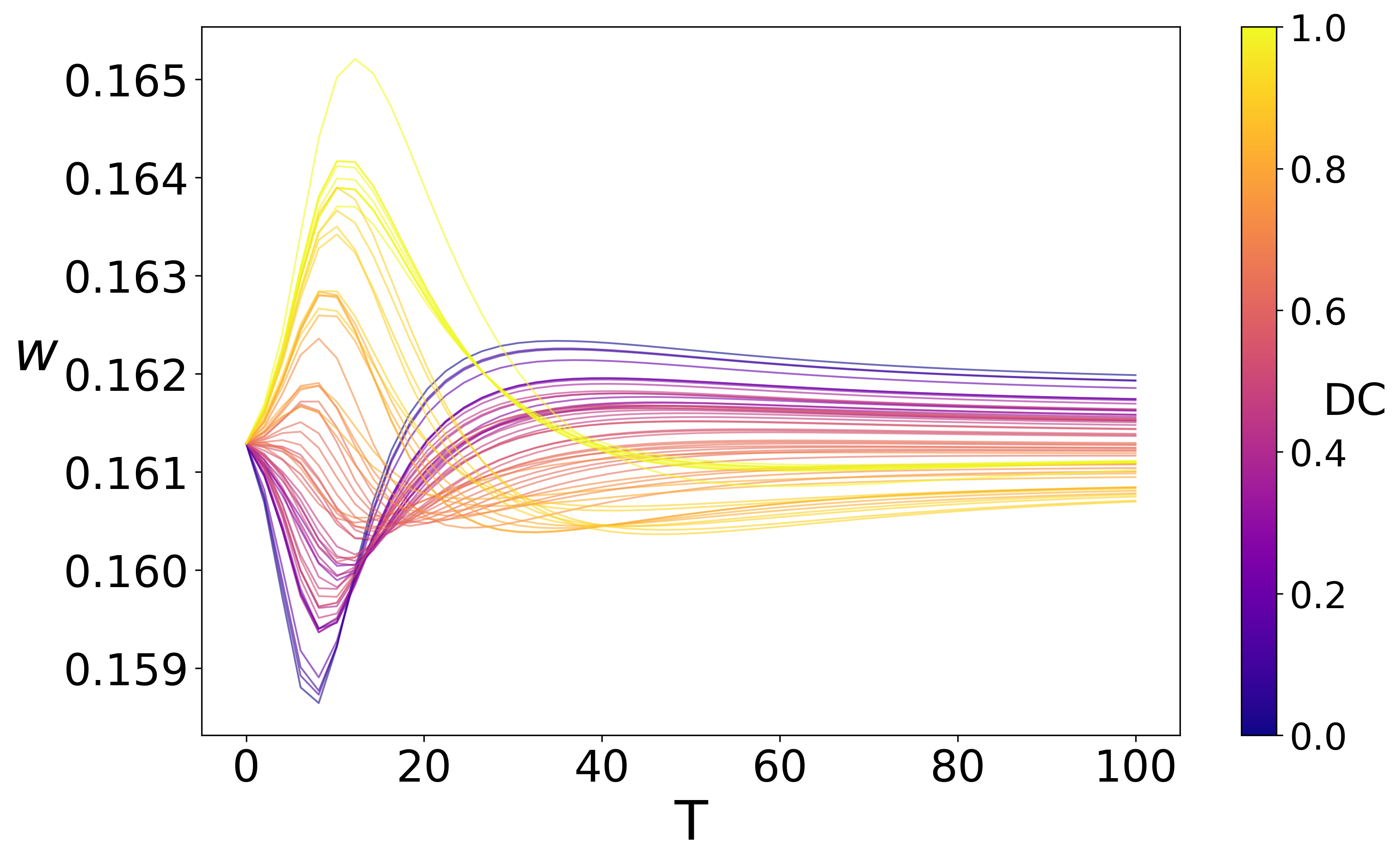}
}\hfill
\subfloat[KT\label{fig:karate}]{
  \includegraphics[width=0.31\textwidth]{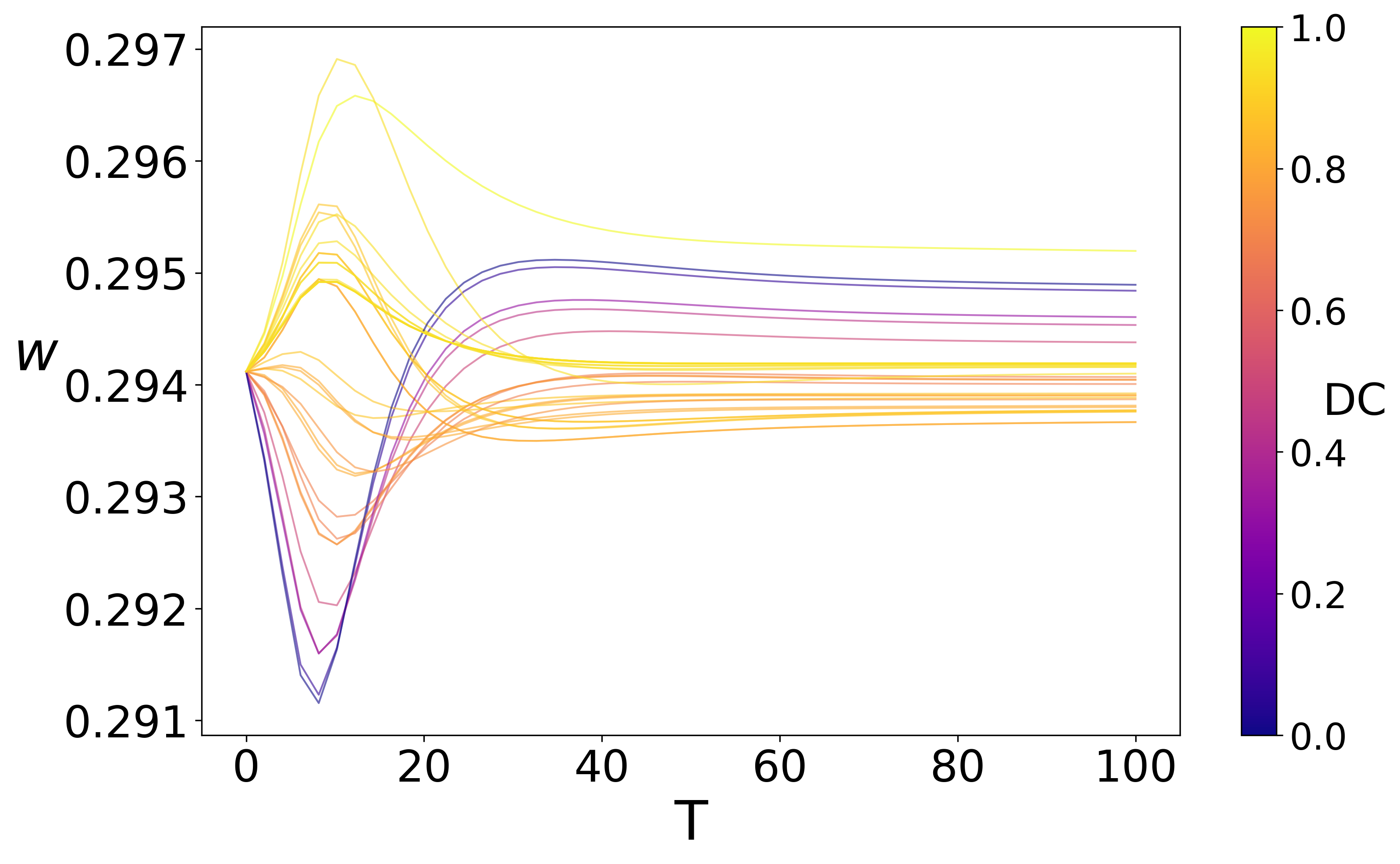}
}\hfill
\subfloat[EM\label{fig:email}]{
  \includegraphics[width=0.31\textwidth]{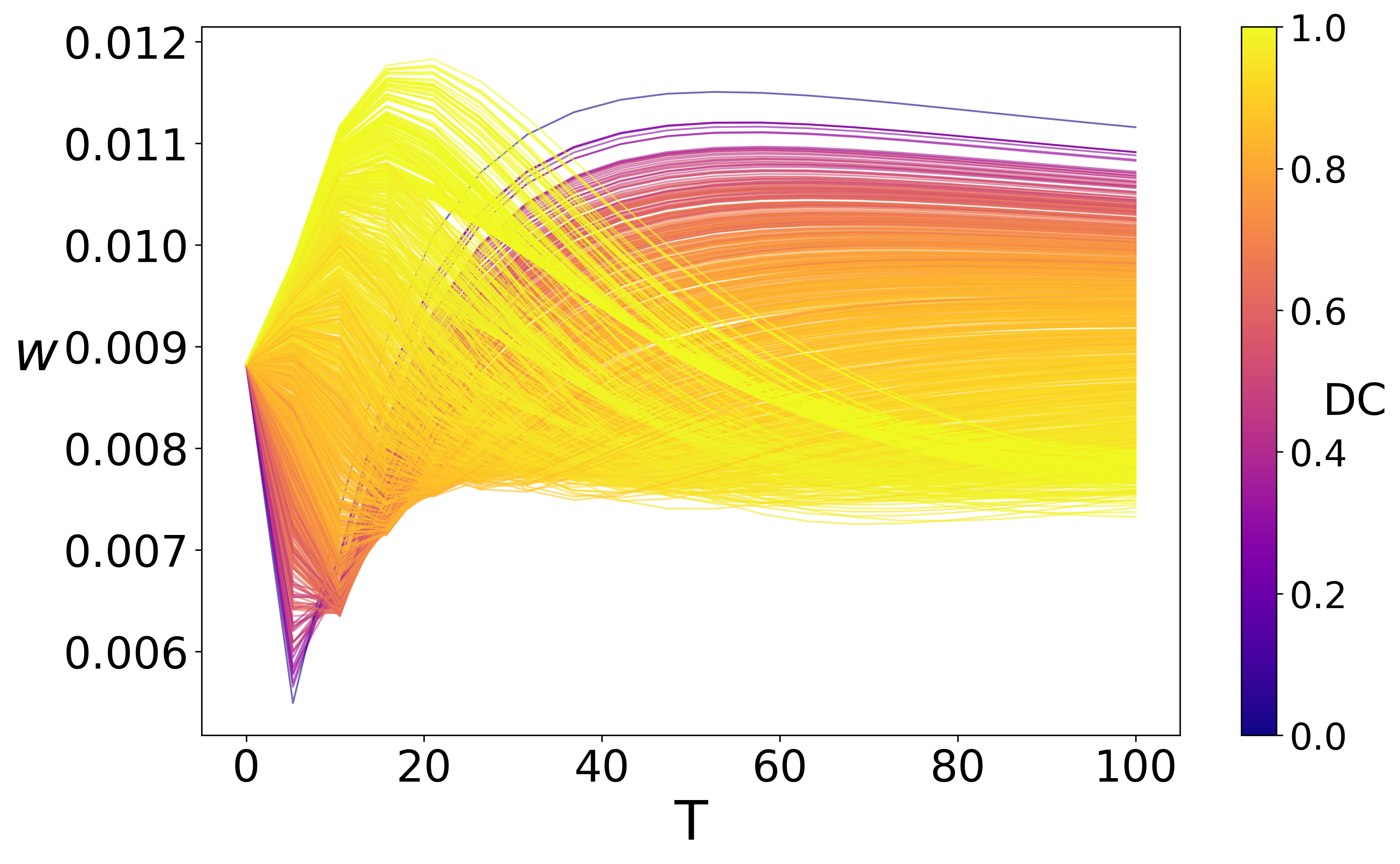}
}

\par\medskip % ← 换行，开始第二行

% Row 2
\subfloat[SC\label{fig:sociopatterns}]{
  \includegraphics[width=0.31\textwidth]{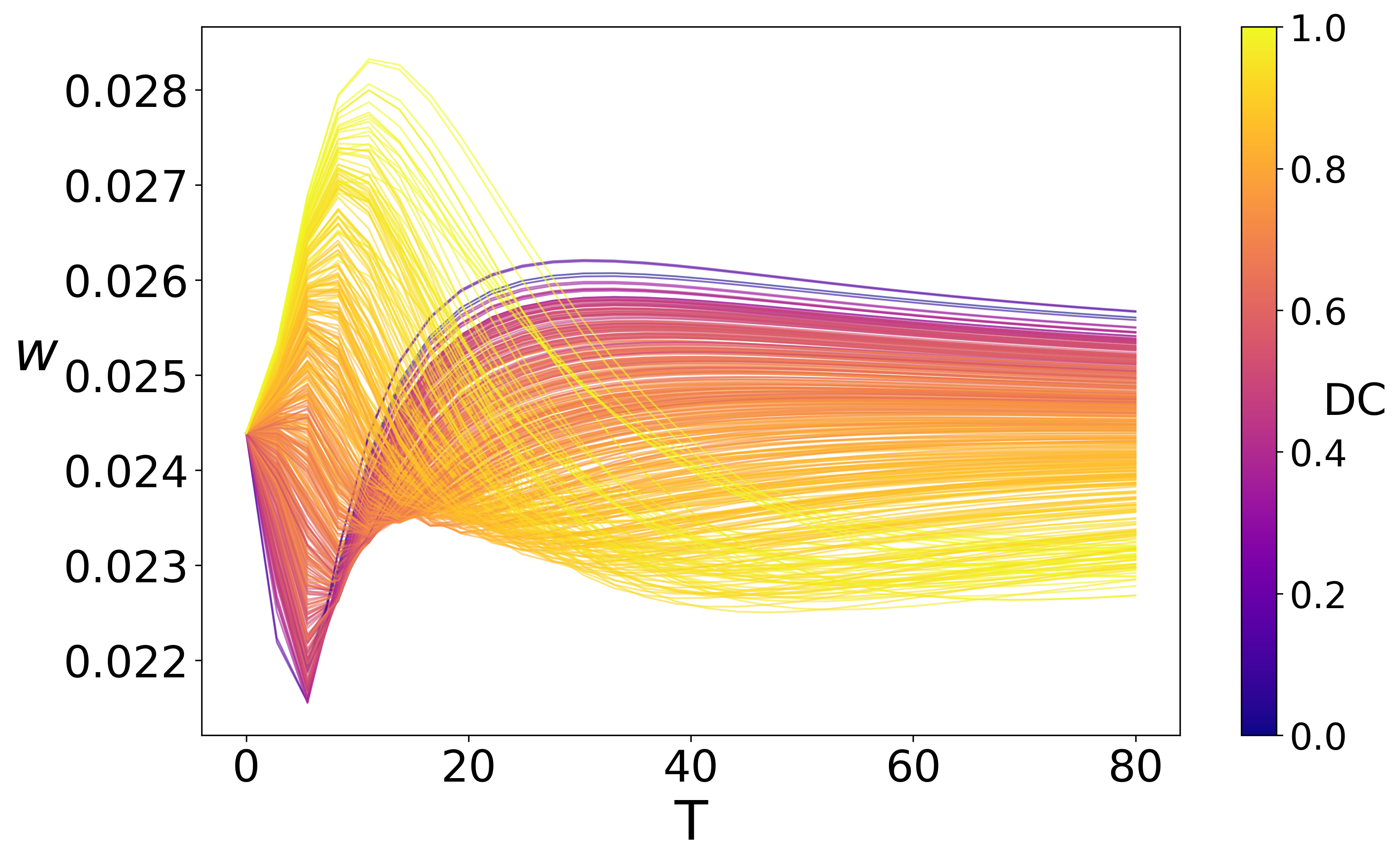}
}\hfill
\subfloat[BA\label{fig:ba}]{
  \includegraphics[width=0.31\textwidth]{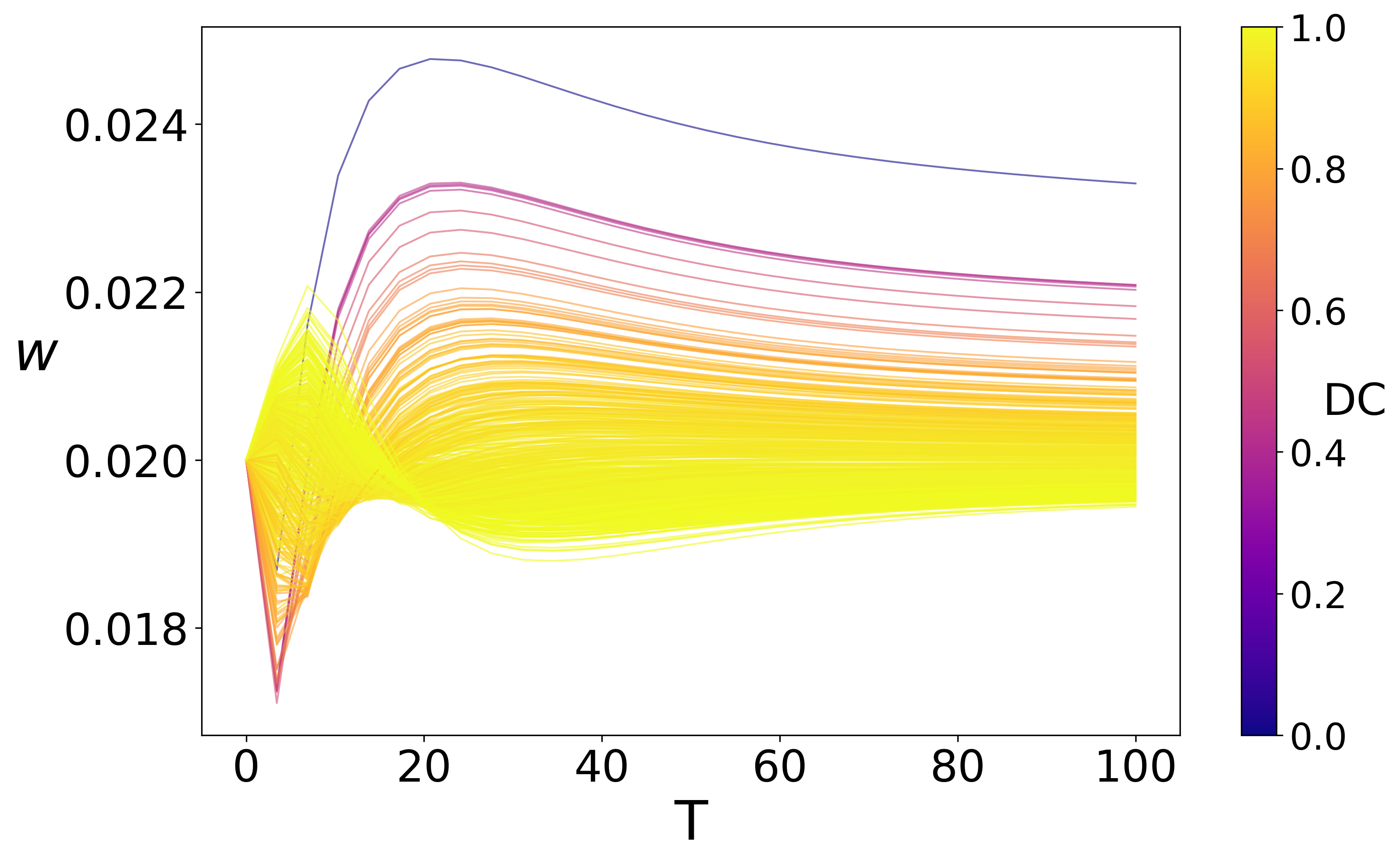}
}\hfill
\subfloat[WS\label{fig:ws}]{
  \includegraphics[width=0.31\textwidth]{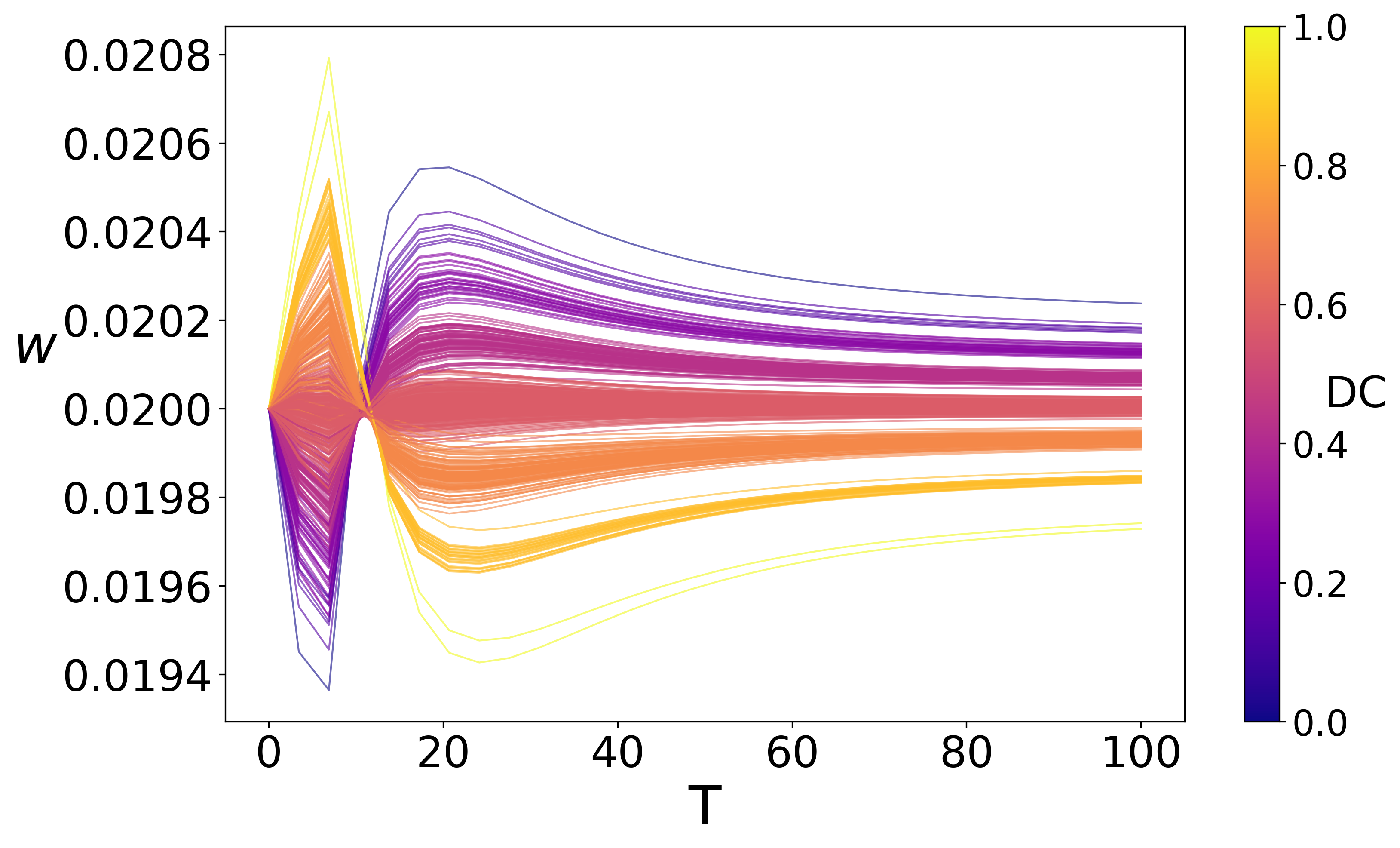}
}

\caption{The weight distribution of nodes with different degrees in different networks varies over time.}
\label{fig:grid}
\end{figure*}

Figure~\ref{fig:grid} shows how control weights $w$ are distributed among nodes of different degrees during rumor propagation across various networks. The control weights exhibit distinct trends depending on the node's structural role within the network.

In the early stage, high-degree nodes receive relatively higher control weights. This indicates that the strategy first prioritizes suppressing central nodes, preventing them from becoming persistent local sources of rumor spreading. As time progresses into the intermediate stage, the weight allocation becomes more balanced across degrees. This reflects a transition where both low- and high-degree nodes are controlled simultaneously, limiting both local outbreaks and the risk of rapid diffusion through hubs. In the late stage, low-degree nodes dominate the allocation. Concentrating resources on  peripheral nodes ensures that the main diffusion channels remain suppressed, preventing the rumor from resurging. Across all network types, this stage-specific pattern is consistent.

Overall, the results indicates that the control weights obtained through the optimal control theory is related to the centrality of the node and is related to the structure and position of the node in the network.

\subsection{\label{Global Infection}Global Infection I Density Curve Variation}

The third experiment evaluates the effectiveness of the proposed strategy at the global level by simulating rumor propagation under the controlled SIR dynamics. We simulate the controlled-SIR model using Algorithm~\ref{alg:generic-fbs}. Two important aspects need to be considered during the simulation:
\begin{enumerate}
    \item The first objective is to assess whether the application of the control weight \(w_i\) at each node results in a more effective reduction of the infection density \(I(t)\) when compared to the uncontrolled scenario, where no control is applied.
    
    \item We will compare the performance of different control strategies in minimizing the cumulative total cost while effectively reducing the infection density.
\end{enumerate}

We compare the infection density curves produced by the optimal dynamic strategy with those under several static allocation baselines, including uniform allocation and centrality-based allocations, positively or negatively correlated with DC, BC, and CC, as well as cycle number(CN)~\cite{shi2023cost} and cycle ratio(CR)~\cite{fan2021characterizing}. 

% To validate the effectiveness of the proposed optimal control strategy on a global scale, a series of dynamic experiments are conducted to assess its overall suppression effect on the rumor propagation process. The objective of these experiments is to simulate rumor propagation under different control strategies and compare their performance in reducing the number of infected individuals, limiting the spread, and optimizing resource allocation.

Among them, the symbol convention clarifies the design of the baseline:
\begin{itemize}
  \item \textbf{DC$+$} and \textbf{DC$-$} denote weight allocations positively and negatively correlated with degree centrality, respectively.
  \item \textbf{BC$+$} and \textbf{BC$-$} denote weight allocations positively and negatively correlated with betweenness centrality, respectively.
  \item \textbf{CC$+$} and \textbf{CC$-$} denote weight allocations positively and negatively correlated with closeness centrality, respectively.
  \item \textbf{CN$+$} and \textbf{CN$-$} denote weight allocations positively and negatively correlated with cycle number, respectively.
  \item \textbf{CR$+$} and \textbf{CR$-$} denote weight allocations positively and negatively correlated with cycle ratio, respectively.
  \item \textbf{UN} denotes that each node is evenly allocated weights.

  \item \textbf{DRA} (Priority Planning) is a dynamic baseline that uses a priority order $\pi$ (smaller index = higher priority) learned by approximately minimizing the order’s maxcut~\cite{scaman2016suppressing}. At each segment, the entire intervention budget is assigned first to the currently infected nodes with the highest priority in $\pi$, resources flow to priority nodes before any lower-priority ones. In our SIR instantiation, treatment raises recovery while keeping infection fixed, $\gamma_i(t)=\gamma_0+u$ and $\sum_i w_i(t)=1$.

  \item \textbf{Unc} represents the baseline scenario where no control measures are applied, allowing the epidemic to propagate freely across the network.
\end{itemize}

Figure~\ref{fig:sir} compares the infection density trajectories under different control strategies across multiple real-world and synthetic networks. The results consistently demonstrate that the optimal strategy, where control weights are dynamically and adaptively updated at each time step, achieves superior epidemic suppression compared to static baselines. In all cases, the optimal curve results in a lower infection density, which is consistent with our goal of minimizing the total area constraint I, and at the same time, the peak of infection is also at a lower level. We speculate that due to the lower heterogeneity in WS and ER networks, the performance of the optimal strategy is less pronounced compared to the BA network, where the heterogeneity of the network structure allows the control strategy to more effectively target key nodes, thereby enhancing its impact on epidemic suppression.

Here, we set the recovery probability as recovery rate $\gamma_0 = 0.1$, and consider multiple values of the infection probability $\gamma$, with the condition that $\gamma > \beta_c$. The value of $\beta_c$ is the propagation threshold, given by $\beta_c = \frac{\langle k \rangle}{\langle k^2 \rangle - \langle k \rangle}$, where $\langle k \rangle$ and $\langle k^2 \rangle$ are the average degree and the average squared degree, respectively~\cite{pastor2015epidemic}.

\begin{figure}[H]
\centering

% Row 1
\subfloat[BA\label{fig:ba_sir}]{
  \includegraphics[width=0.31\textwidth]{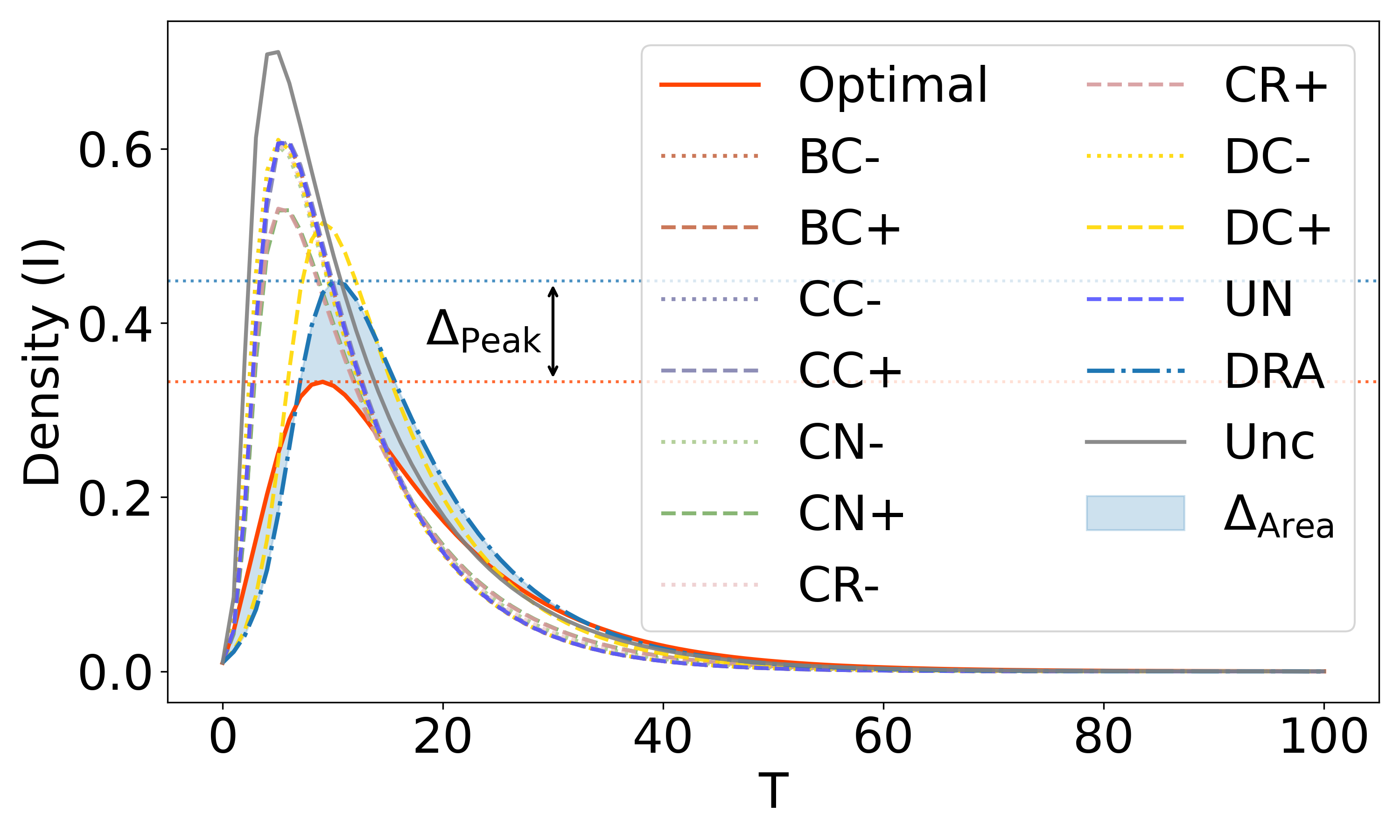}
}\hfill
\subfloat[WS\label{fig:ws_sir}]{
  \includegraphics[width=0.31\textwidth]{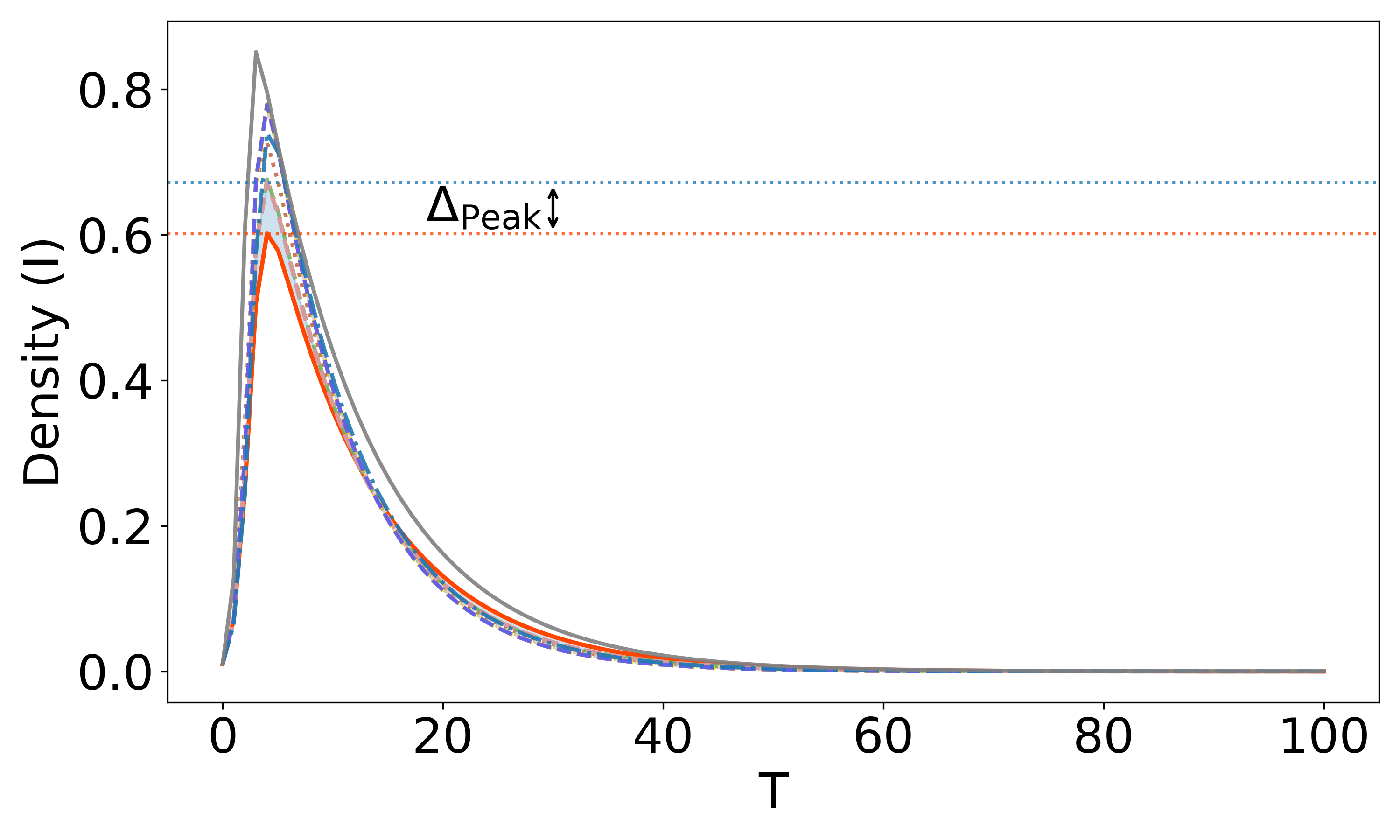}
}\hfill
\subfloat[ER\label{fig:er_sir}]{
  \includegraphics[width=0.31\textwidth]{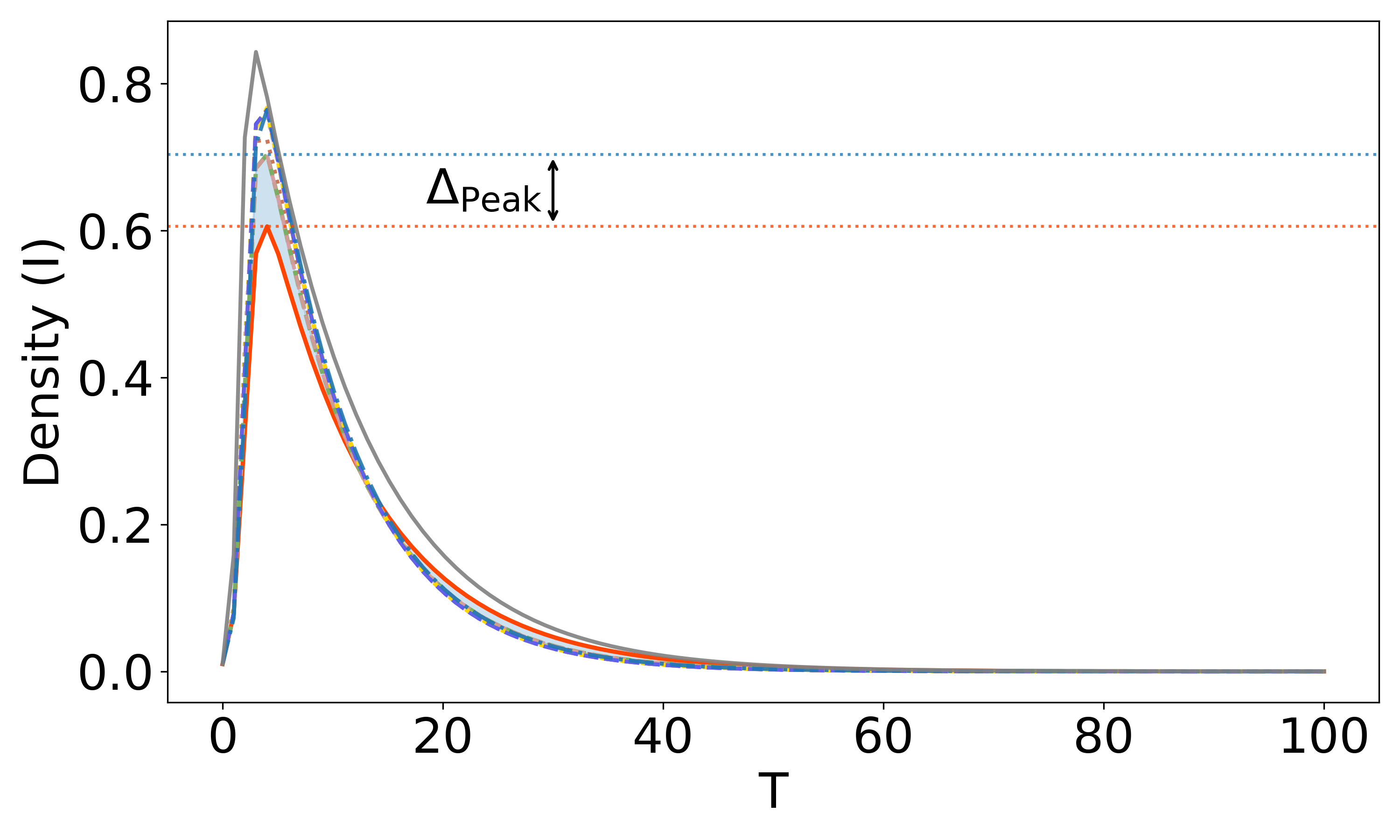}
}

\par\medskip % ← 换行，开始第二行

\subfloat[FB\label{fig:fb_sir}]{
  \includegraphics[width=0.31\textwidth]{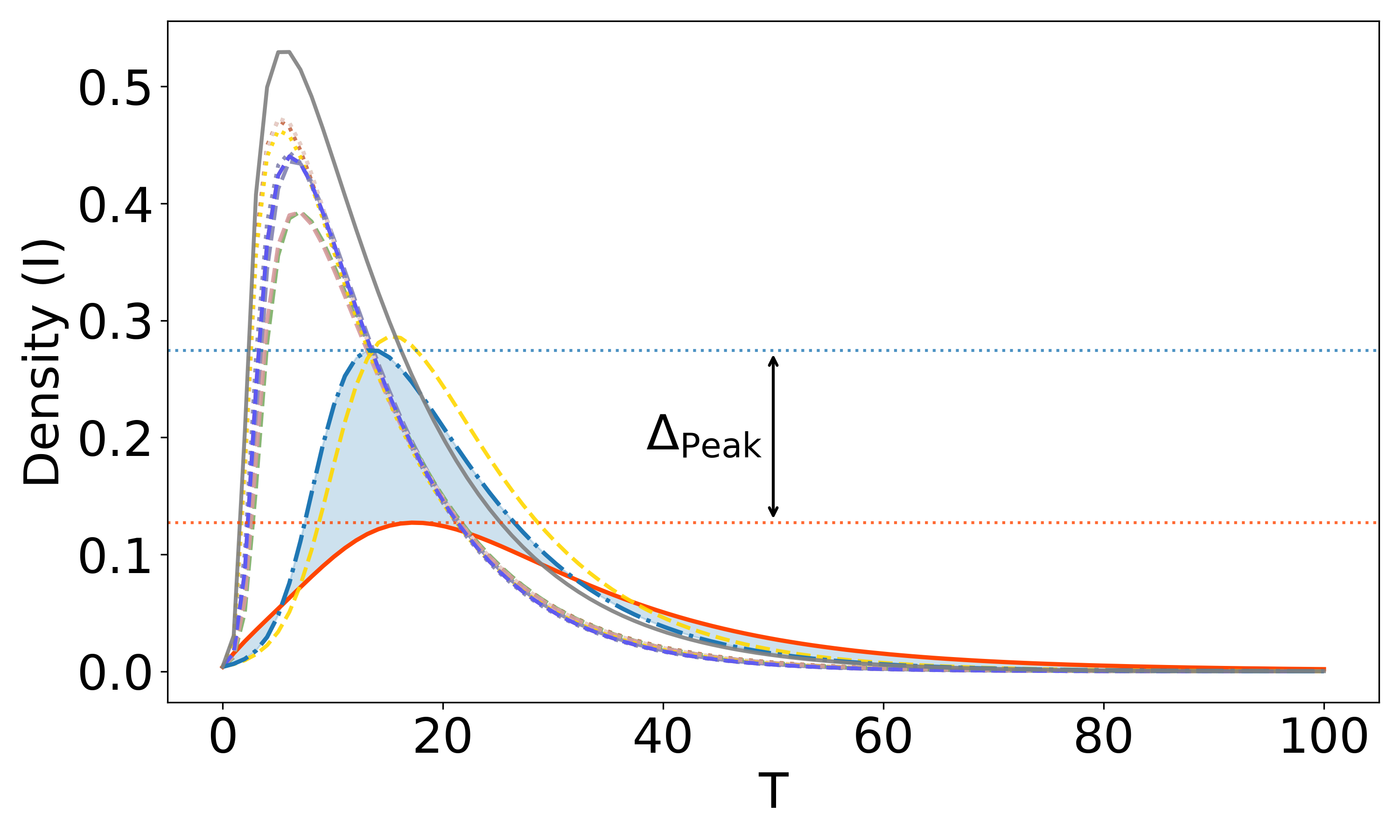}
}\hfill
\subfloat[SC\label{fig:sc_sir}]{
  \includegraphics[width=0.31\textwidth]{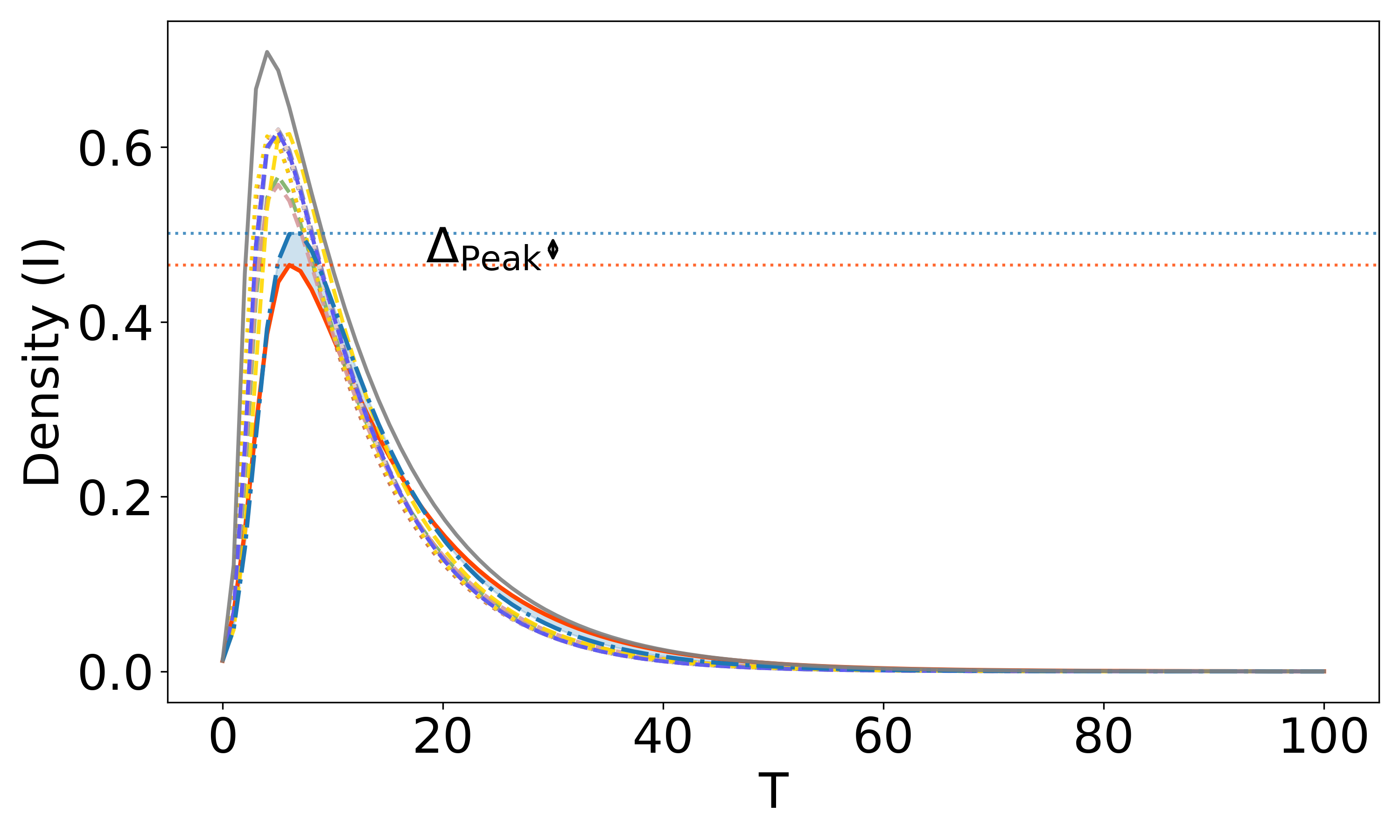}
}\hfill
\subfloat[EM\label{fig:em_sir}]{
  \includegraphics[width=0.31\textwidth]{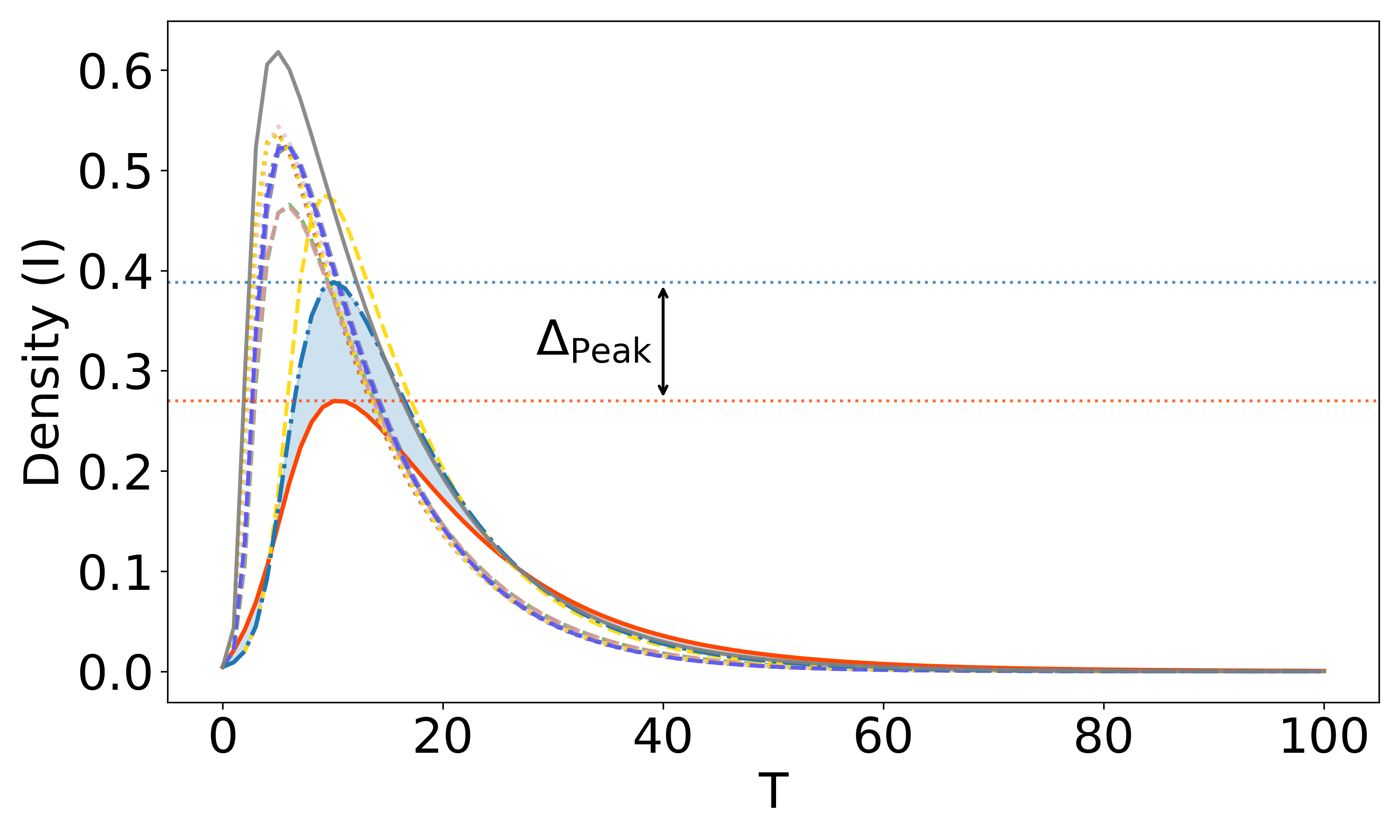}
}

\par\medskip % ← 换行，开始第三行

\subfloat[DG\label{fig:DG_sir}]{
  \includegraphics[width=0.31\textwidth]{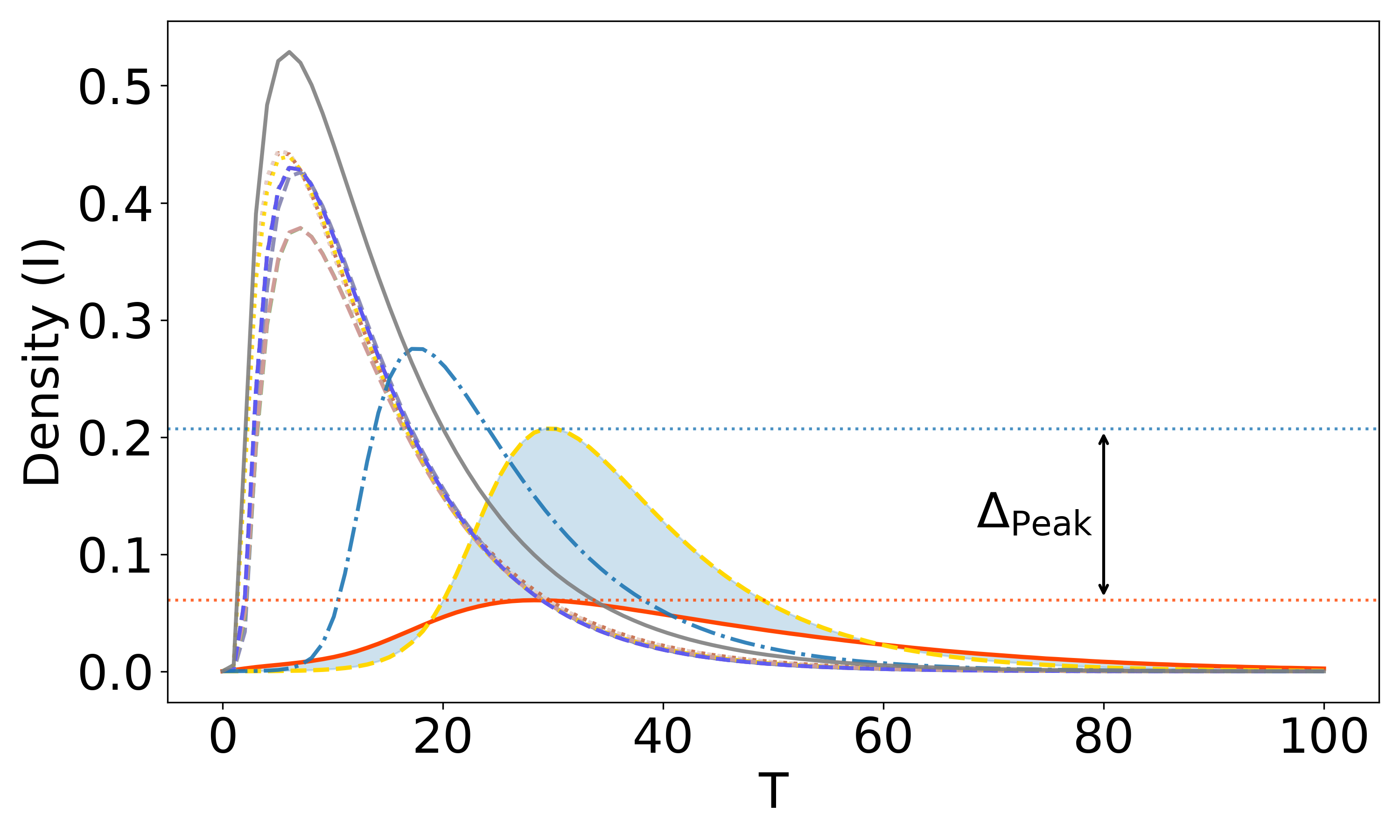}
}\hfill
\subfloat[EN\label{fig:EN_sir}]{
  \includegraphics[width=0.31\textwidth]{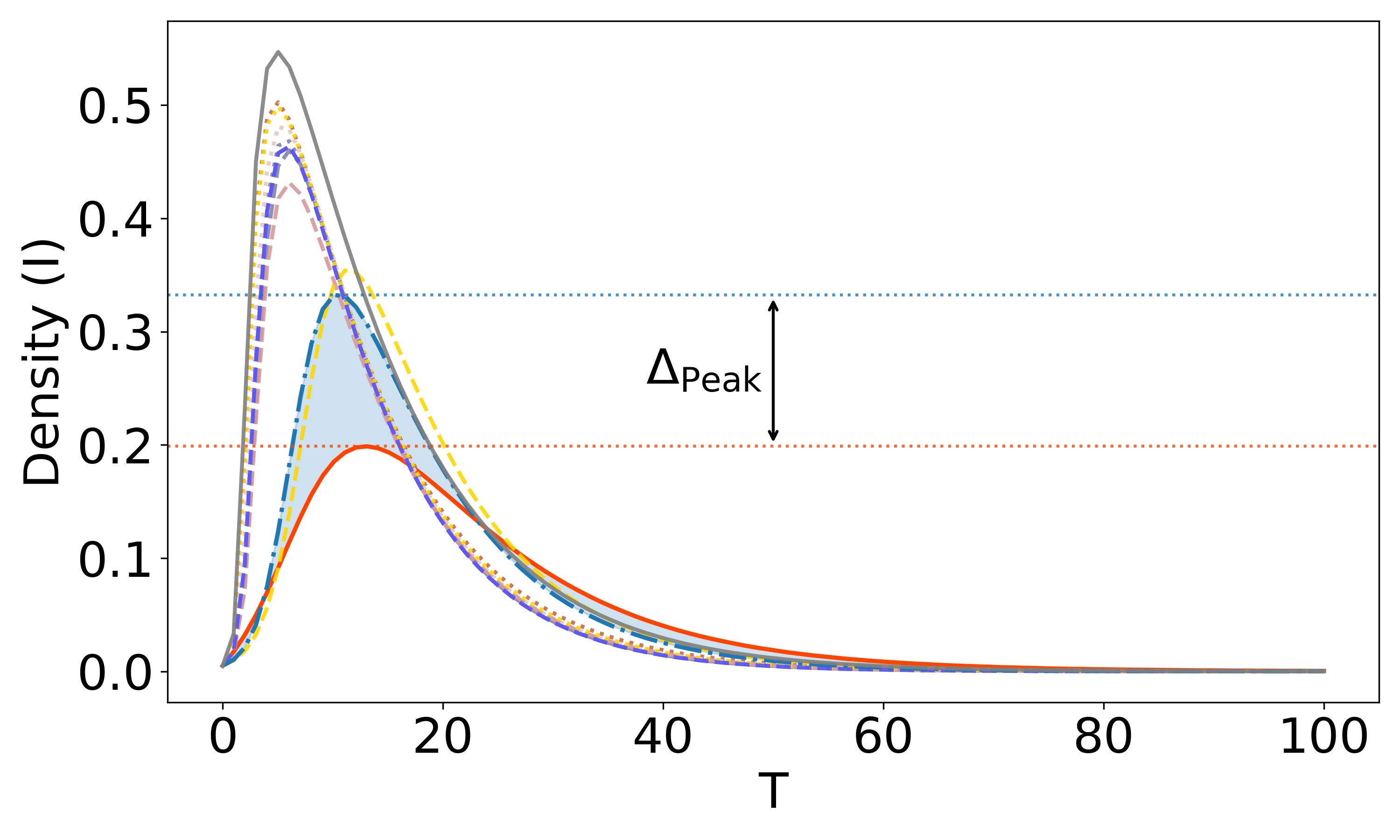}
}\hfill
\subfloat[BL\label{fig:BL_sir}]{
  \includegraphics[width=0.31\textwidth]{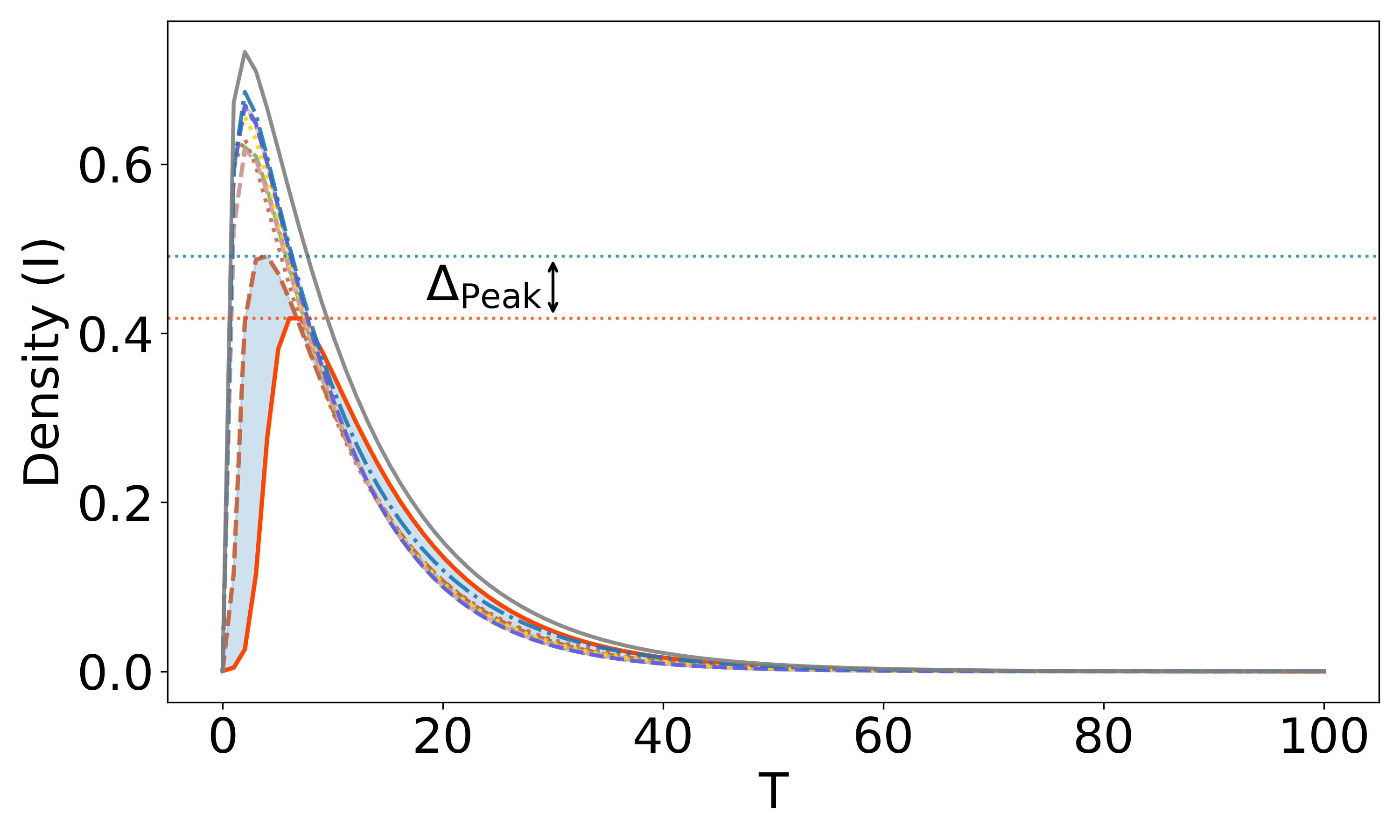}
}

\caption{Different network comparison of infection trajectories: optimal time-varying control vs. centrality-based static policies. These results are captured under SIR model parameters $\gamma_0 = 0.1$, $\beta_0 = 3\beta_c$.}

\label{fig:sir}
\end{figure}

Table~\ref{tab:metrics_all} and further quantify these results using Peak and Area metrics. The \textit{Peak} represents the maximum fraction of infected nodes during the spreading process, while the \textit{Area} denotes the cumulative infection over time, the integral of the infection curve \(I(t)\) with respect to time \(t\):
\begin{equation}
\mathrm{Area} = \int_{0}^{T} I(t)\, \mathrm{d}t ,
\end{equation}
where \(I(t)\) is the total proportion of infected nodes at time \(t\), and \(T\) is the total simulation time. The optimal control strategy achieves the lowest Area values across all networks, confirming that the adaptive allocation of resources dynamically reduces the overall infection burden. Static strategies, while able to slightly reduce the Peak in some cases, fail to minimize the cumulative infection, as reflected in consistently larger Area values. As shown in the Figures~\ref{fig:sir}, we highlight the differences between the optimal and second-best methods in terms of the change values \(\Delta_{\rm{Peak}}\) and \(\Delta_{\rm{Area}}\). Here, \(\Delta\) denotes \(\Delta_{\rm{Peak}}\) or \(\Delta_{\rm{Area}}\) (\(v_{(2)} - v_{(1)}\)), where \(v_{(1)}\) is the minimum value of a given metric among all methods, and \(v_{(2)}\) is the second smallest value. These annotations visually emphasize the relative improvements achieved by the proposed Optimal strategy in suppressing rumor propagation intensity (Peak) and cumulative spread (Area) across diverse network topologies.

\begin{table}[htbp]
\centering
\caption{Peak and Area metrics for different methods on BA, WS, ER, FB, SC, EM, DG, EN, and BL networks. Lower is better.}
\label{tab:metrics_all}

% 紧凑一点的列间距
\setlength{\tabcolsep}{2pt}
\renewcommand{\arraystretch}{1.05}

% 关键：按单栏宽度缩放
\resizebox{\linewidth}{!}{%
\large
\begin{tabular}{l *{18}{c}}
\toprule
& \multicolumn{2}{c}{BA} & \multicolumn{2}{c}{WS} & \multicolumn{2}{c}{ER} & \multicolumn{2}{c}{FB} & \multicolumn{2}{c}{SC} & \multicolumn{2}{c}{EM} & \multicolumn{2}{c}{DG} & \multicolumn{2}{c}{EN} & \multicolumn{2}{c}{BL} \\
\cmidrule(lr){2-3} \cmidrule(lr){4-5} \cmidrule(lr){6-7} \cmidrule(lr){8-9} \cmidrule(lr){10-11} \cmidrule(lr){12-13} \cmidrule(lr){14-15} \cmidrule(lr){16-17} \cmidrule(lr){18-19}
Method & Peak & Area & Peak & Area & Peak & Area & Peak & Area & Peak & Area & Peak & Area & Peak & Area & Peak & Area & Peak & Area \\
\midrule
Optimal & \textbf{0.4681} & \textbf{7.5270} & \textbf{0.6015} & \textbf{7.5932} & \textbf{0.6056} & \textbf{7.6158} & \textbf{0.1297} & \textbf{4.3059} & \textbf{0.4681} & \textbf{7.5270} & \textbf{0.2719} & \textbf{5.9480} & \textbf{0.0608} & \textbf{2.5415} & \textbf{0.1987} & \textbf{5.0995} & \textbf{0.4177} & \textbf{6.1891} \\
UN     & 0.6172 & 7.8651 & 0.7788 & 7.9983 & 0.7639 & 7.9984 & 0.4403 & 6.8307 & 0.6172 & 7.8651 & 0.5251 & 7.4543 & 0.4298 & 6.9222 & 0.4634 & 6.6457 & 0.6694 & 7.6884 \\
DC$+$    & 0.6210 & 8.0124 & 0.7794 & 8.0014 & 0.7675 & 8.0111 & 0.2895 & 6.1071 & 0.6210 & 8.0124 & 0.4767 & 7.4685 & 0.2072 & 4.9308 & 0.3539 & 6.2844 & 0.6695 & 7.6895 \\
DC$-$  & 0.6136 & 7.8350 & 0.7780 & 8.0014 & 0.7589 & 8.0134 & 0.4626 & 7.1504 & 0.6136 & 7.8350 & 0.5363 & 7.5592 & 0.4397 & 7.1337 & 0.4986 & 7.2252 & 0.6562 & 7.7198 \\
BC$+$     & 0.6173 & 7.8647 & 0.7789 & 7.9981 & 0.7639 & 7.9983 & 0.4402 & 6.8302 & 0.6173 & 7.8647 & 0.5251 & 7.4541 & 0.4298 & 6.9222 & 0.4633 & 6.6447 & 0.4911 & 6.5778 \\
BC$-$     & 0.6137 & 7.7055 & 0.7267 & 8.0047 & 0.7243 & 7.9874 & 0.4715 & 7.3354 & 0.6137 & 7.7055 & 0.5373 & 7.4978 & 0.4423 & 7.2910 & 0.5024 & 7.3405 & 0.6300 & 7.5584 \\
CC$+$     & 0.6201 & 7.8802 & 0.7791 & 7.9991 & 0.7646 & 7.9988 & 0.4366 & 6.8198 & 0.6201 & 7.8802 & 0.5254 & 7.4655 & 0.4259 & 6.8559 & 0.4600 & 6.5901 & 0.6708 & 7.6939 \\
CC$-$     & 0.6136 & 7.8598 & 0.7784 & 7.9990 & 0.7631 & 7.9987 & 0.4437 & 6.8516 & 0.6136 & 7.8598 & 0.5252 & 7.4508 & 0.4308 & 6.9415 & 0.4681 & 6.7508 & 0.6679 & 7.6871 \\
CN$+$     & 0.5662 & 7.5866 & 0.6759 & 7.6998 & 0.7035 & 7.7515 & 0.3933 & 6.4777 & 0.5662 & 7.5866 & 0.4658 & 7.1025 & 0.3784 & 6.4030 & 0.4270 & 6.4092 & 0.6201 & 7.4709 \\
CN$-$     & 0.6203 & 7.8648 & 0.7754 & 8.0823 & 0.7573 & 8.0639 & 0.4729 & 7.3372 & 0.6203 & 7.8648 & 0.5438 & 7.5452 & 0.4442 & 7.2359 & 0.4807 & 6.9617 & 0.6722 & 7.7226 \\
CR$+$     & 0.5569 & 7.5666 & 0.6716 & 7.6727 & 0.7031 & 7.7538 & 0.3930 & 6.4838 & 0.5569 & 7.5666 & 0.4637 & 7.1031 & 0.3785 & 6.4078 & 0.4314 & 6.4009 & 0.6156 & 7.4234 \\
CR$-$   & 0.6203 & 7.8648 & 0.7745 & 8.0969 & 0.7610 & 8.0548 & 0.4729 & 7.3372 & 0.6203 & 7.8648 & 0.5438 & 7.5452 & 0.4442 & 7.2359 & 0.4807 & 6.9617 & 0.6722 & 7.7226 \\
DRA     & 0.5096 & 7.5474 & 0.7395 & 8.1212 & 0.7645 & 8.1112 & 0.2766 & 5.7733 & 0.5096 & 7.5474 & 0.3889 & 6.8441 & 0.2753 & 5.4633 & 0.3322 & 5.9707 & 0.6854 & 8.3110 \\
Unc     & 0.7087 & 9.9374 & 0.8511 & 9.9990 & 0.8430 & 9.9976 & 0.5305 & 9.2408 & 0.7087 & 9.9374 & 0.6180 & 9.7141 & 0.5287 & 9.3441 & 0.5468 & 8.8115 & 0.7326 & 9.7632 \\
\midrule
\(P\) (\%) & 8.87 & 0.27 & 11.65 & 1.05 & 16.10 & 1.78 & 113.26 & 34.08 & 8.87 & 0.27 & 43.03 & 15.07 & 240.79 & 94.01 & 67.19 & 17.08 & 17.57 & 6.28 \\
\bottomrule
\end{tabular}
}
\end{table}

As shown in the last row of Table~\ref{tab:metrics_all}, the suppression efficiency \(P\) (\%) quantifies the relative improvement of the Optimal method over the second-best baseline across nine representative network topologies. The two evaluation metrics, \textit{Peak} and \textit{Area}, both measure the extent of rumor propagation, where smaller values indicate stronger suppression. The efficiency values in all networks are positive, demonstrating that our method consistently achieves lower infection peaks and smaller overall outbreak sizes than competing methods. For instance, in the DG network, the \textit{Peak} suppression efficiency reaches 240.79\%, while the \textit{Area} efficiency is 94.01\%. These results indicate that the proposed control strategy effectively curtails both the intensity and duration of rumor spreading across diverse network structures.

The suppression efficiency \(P\) is computed as
\begin{equation}
P = \frac{\Delta}{v_{(2)}} \times 100\% ,
\end{equation}
where \(\Delta\) denotes \(\Delta_{\rm{Peak}}\) or \(\Delta_{\rm{Area}}\). Therefore, higher \(P\) values correspond to greater relative improvement in rumor suppression performance.

Overall, both curve visualization and quantitative indicators consistently demonstrate that the proposed optimal dynamic control strategy achieves the most significant reduction of infection area across all tested networks. Unlike static allocation methods that either reduce the infection peak without affecting the long-term prevalence, or slightly delay the outbreak while still incurring a large cumulative burden, the optimal strategy simultaneously lowers the peak and accelerates the decline phase, leading to a substantial decrease in total infection exposure over time. This outcome directly validates the effectiveness of our problem formulation, in which the cost function explicitly targets the minimization of cumulative infection.

Furthermore, the consistency between the curve shapes and the quantitative metrics in Table~\ref{tab:metrics_all} highlights the robustness of the proposed approach. Even when static strategies perform competitively in peak suppression for specific networks, they fail to match the optimal method in terms of overall area reduction. By dynamically reallocating resources according to the evolving state of the system, the optimal control adapts to different stages of rumor diffusion, ensuring that resources are used efficiently throughout the process rather than concentrated at a single phase.

These results demonstrate that the proposed node-level adaptive optimal control strategy is both theoretically sound and practically feasible.    By dynamically reallocating resources at each step of rumor diffusion, the method ensures efficient use of limited interventions and consistently achieves the smallest infection area, directly validating our formulation in realistic network settings. The incorporation of a quadratic control cost automatically balances prevalence reduction and resource expenditure, penalizing excessive allocation while maintaining strong suppression.    Beyond rumor containment, this principle of adaptive, time-varying resource allocation can be extended to other domains such as epidemic control, information dissemination, cybersecurity, and infrastructure protection, where achieving an effective balance between suppression effectiveness and resource efficiency is equally critical.

\section{\label{Conclusion}Conclusion} 

This work develops a structure-aware, resource-constrained optimal control framework for rumor containment on networks. The controller operates at the node level with time-varying intervention weights derived from a principled optimization formulation. We establish the existence of an optimal solution and implement an efficient forward--backward sweep procedure to compute policies. 

Extensive simulations on synthetic and real networks show that the proposed approach consistently reduces both peak prevalence and cumulative infection area relative to uniform and centrality-based static baselines. A robust stage-aware allocation law emerges across settings: early suppression of influential hubs followed by targeted cleanup on peripheral nodes. This law provides an interpretable operational rule that balances effectiveness with resource usage. 

The framework links policy design to network topology, yielding a scalable and transparent tool for managing information diffusion. Beyond rumor mitigation, the same principles apply to misinformation management, public-health communication, and crisis response on large-scale platforms. 

This study has limitations. The analysis assumes full-state observability and fixed budgets without delays; model mismatch and partial observations are only indirectly addressed. Future work will strengthen theoretical guarantees (convergence and approximation bounds), incorporate partial observability and budget uncertainty, and extend the framework to multilayer dynamics, time-delay effects, and fairness-aware constraints.

% or


% \appendices
% % \section{Proof of the First Zonklar Equation}
% % Appendix one text goes here.

% % you can choose not to have a title for an appendix
% % if you want by leaving the argument blank
% \section{}

\section*{Acknowledgment}
This study was supported by the National Natural Science Foundation of China (Grant Nos. T2293771, 62503447), the STI 2030 Major Projects (Grant No. 2022ZD0211400), the China Postdoctoral Science Foundation (Grant No. 2024M763131), the Postdoctoral Fellowship Program of CPSF (Grant No. GZC20241653), and the New Cornerstone Science Foundation through the XPLORER PRIZE.

\section*{Credit authorship contribution statement}
Yan Zhu: Conceptualization, Methodology, Software, Writing -- original draft. Qingyang Liu: Conceptualization, Methodology, Software, Validation. Chang Guo: Methodology, Validation, Writing -- review \& editing. Tianlong Fan: Conceptualization, Methodology, Validation, Writing -- review \& editing, Supervision, Funding acquisition. Linyuan L\"u: Conceptualization, Writing -- review \& editing, Supervision, Funding acquisition.

\section*{Declaration of competing interest}
The authors declare no competing interests.

\section*{Data and code availability}
Data and code will be made available on request.

%%       \cite{<label>} ==> [1]




% %% For numbered reference style
%   \textit{\LaTeX: a document preparation system},
%   Addison Wesley, Massachusetts,
% \end{thebibliography}

\bibliographystyle{elsarticle-num}             % 数字制
\bibliography{main}  
\end{document}